%% file: ElementarymatrixZX.tex
\documentclass[submission]{eptcs}
 \usepackage{amssymb,nicefrac}
\usepackage{amsfonts}
\usepackage{graphicx}
\usepackage[fleqn]{amsmath}
\usepackage[varg]{txfonts}
\usepackage{stmaryrd}

\usepackage{framed}
\usepackage{color}
\usepackage[normalem]{ulem}
\usepackage{tikzfig}
\usepackage{placeins}

\usepackage{blkarray}

\usepackage{amsthm}

\usepackage{thmtools}
\usepackage{thm-restate}

\usepackage{hyperref}
\usepackage{cleveref}
\usepackage[
  backend=biber,
]{biblatex}
\newcommand{\doi}[1]{\url{http://dx.doi.org/#1}}
\usepackage{mathtools} 
\DeclarePairedDelimiter\bra{\langle}{\rvert}
\DeclarePairedDelimiter\ket{\lvert}{\rangle}
\input{tikzfigures.tex}
\input{tikzstyles.tex}

\input{defs.tex}

\newtheorem{Th}{Theorem}[section]
\newtheorem{theorem}[Th]{Theorem}
\newtheorem{proposition}[Th]{Proposition} 
\newtheorem{lemma}[Th]{Lemma}

\newtheorem{example}[Th]{Example}
\newtheorem{remark}[Th]{Remark}

\makeatletter
\newcommand{\vast}{\bBigg@{6.5}}
\makeatother

\newcommand{\vertrule}[1][1ex]{\rule{.4pt}{#1}}

\title{Representing and Implementing Matrices Using Algebraic ZX-calculus}
\author{Quanlong Wang$^{1}$ \quad\quad Richie Yeung$^{1,2}$\\
$^{1}$Quantinuum, 17 Beaumont Street, Oxford OX1 2NA, United Kingdom\\ $^{2}$University of Oxford}

\def\titlerunning{Representing and Implementing Matrices Using Algebraic ZX-calculus}
\def\authorrunning{Quanlong Wang and Richie Yeung}
\begin{document}
\date{}\maketitle
\begin{abstract}
In linear algebra applications, elementary matrices hold a significant role. This paper presents a diagrammatic representation of all $2^m\times 2^n$-sized elementary matrices in algebraic ZX-calculus, showcasing their properties on inverses and transpose through diagrammatic rewriting. Additionally, the paper uses this representation to depict the Jozsa-style matchgate in algebraic ZX-calculus. To further enhance practical use, we have implemented this representation in \texttt{discopy}. Overall, this work sets the groundwork for more applications of ZX-calculus such as synthesising controlled matrices \cite{shaikhHowSumExponentiate2022} in quantum computing.
\end{abstract}

 
  \section{ Introduction}
 
Matrices are used everywhere in modern science, like machine learning \cite{machunelearn} or quantum computing \cite{Nielsen}, to name a few. Meanwhile, there is a graphical language called ZX-calculus  that could also deal with matrix calculations such as matrix multiplication and tensor product \cite{CoeckeDuncan,Chkw2021}. Then  there naturally arises a  question: why are people bothering with using diagrams for matrix calculations given that matrix technology has been applied with great successes? There are a few reasons for doing so. First, there is a lot of redundancy in matrix calculations which could be avoided in graphical calculus. For example,  to prove the cyclic property of matrices $tr(AB)=tr(BA)$, all the elements of the two matrices will be involved, while in graphical language like ZX-calculus, the proof of  the cyclic property is almost a tautology \cite{Coeckebk}. Second, matrix calculations always have all the elements of matrices involved, thus a ``global'' operation, while in ZX-calculus, the operations are just diagram rewriting where only a part of a diagram is replaced by another sub-diagram  according to certain rewriting rule, thus essentially a ``local" operation which makes things much easier. Finally, graphical calculus is much more intuitive than matrix calculation,  therefore a pattern/structure is more probably to be recognised in a graphical formalism. In fact,  as a graphical calculus for matrix calculation, ZX-calculus has achieved plenty of successes in the field of quantum computing and information \cite{debeaudrapbianwang, Cowtan_2020, KissingerG20, Sivarajah_2020}

For research realm beyond quantum, traditional ZX-calculus \cite{CoeckeDuncan} is inconvenient  as ``it lacks a way to directly encode the complex numbers" \cite{Weteringreviewpaper}.  To remedy this inconvenience while still preserving its power, we introduced an ``extended" version called algebraic ZX-calculus with new generators added in (not a real extension in the sense that  the new generators can be expressed by the original generators) \cite{wangalg2020}. While \cite{qwangnormalformbit} proves the completeness of algebraic ZX calculus using elementary matrix operations, the set of elementary operations given are a minimal set (multiplying the bottom row, adding the bottom row to other rows) and are unsuitable for practical matrix synthesis.


In this paper, we first provide diagrammatic representations for adding, multiplying and switching the rows and columns of a $2^m\times 2^n$ matrix, diagrammatically prove some basic properties about their inverses and transposes.  Then we show how to represent an arbitrary matrix of size $2^m\times 2^n$ in ZX diagrams. Finally we show how this representation method could be used for representing matchgates \cite{jozsa2008matchgates}.
We also include an implementation of this representation method in \texttt{discopy} \href{https://github.com/y-richie-y/qpl-represent-matrix}{(see [link])}, which numerically verifies these examples.
The results of this paper can be generalised from the field of complex numbers to arbitrary commutative semirings due to the completeness of algebraic ZX-calculus over semirings \cite{qwangrsmring}.

\textbf{Related Work:}
There is an established and related body of research known as Graphical Linear Algebra. \cite{bonchi2014interacting}
The diagrams in this formalism with $n$ inputs and $m$ outputs represent an $m \times n$ matrix, whilst the diagrams in ZX-calculus with $n$ inputs and $m$ outputs represent a $2^m \times 2^n$ matrix, which is more compact and more appropriate for representing the states and processes used in quantum computing. The techniques developed in this paper are also distinct from existing ZX publications that use elementary matrices for circuit synthesis \cite{duncan2020graph, kissinger2019cnot} as they either apply only to parity circuits or is part of a circuit extraction routine that apply to the $m \times n$ parity matrix, rather than the unitary matrix itself.




 \section{Algebraic ZX-calculus}
 In this section, we give a brief introduction to algebraic ZX-calculus by showing its generators, standard interpretation and rewriting rules. More details about algebraic ZX-calculus can be found in \cite{wangalg2020} and \cite{qwangnormalformbit}.
 First we give the generators of algebraic ZX-calculus as follows:
   \begin{table}[!h]
\begin{center}
\begin{tabular}{|r@{~}r@{~}c@{~}c|r@{~}r@{~}c@{~}c|}
\hline
$R_{Z,a}^{(n,m)}$&$:$&$n\to m$ & \scalebox{0.70}{%
	\beginpgfgraphicnamed{TikZit//generalgreenspider}
	\InputIfFileExists{TikZit//generalgreenspider.tikz}{}{\input{./figures/TikZit//generalgreenspider.tikz}}%
	\endpgfgraphicnamed
}  &  $\mathbb I$&$:$&$1\to 1$&%
	\beginpgfgraphicnamed{TikZit//Id}
	\InputIfFileExists{TikZit//Id.tikz}{}{\input{./figures/TikZit//Id.tikz}}%
	\endpgfgraphicnamed
 \\\hline
$H$&$:$&$1\to 1$ &\scalebox{0.8}{%
	\beginpgfgraphicnamed{TikZit//newhadamard}
	\InputIfFileExists{TikZit//newhadamard.tikz}{}{\input{./figures/TikZit//newhadamard.tikz}}%
	\endpgfgraphicnamed
}&  $\sigma$&$:$&$ 2\to 2$& %
	\beginpgfgraphicnamed{TikZit//swap}
	\InputIfFileExists{TikZit//swap.tikz}{}{\input{./figures/TikZit//swap.tikz}}%
	\endpgfgraphicnamed
\\\hline
   $C_a$&$:$&$ 0\to 2$& %
	\beginpgfgraphicnamed{TikZit//cap}
	\InputIfFileExists{TikZit//cap.tikz}{}{\input{./figures/TikZit//cap.tikz}}%
	\endpgfgraphicnamed
 &$ C_u$&$:$&$ 2\to 0$&%
	\beginpgfgraphicnamed{TikZit//cup}
	\InputIfFileExists{TikZit//cup.tikz}{}{\input{./figures/TikZit//cup.tikz}}%
	\endpgfgraphicnamed
 \\\hline
  $T$&$:$&$1\to 1$&%
	\beginpgfgraphicnamed{TikZit//triangle}
	\InputIfFileExists{TikZit//triangle.tikz}{}{\input{./figures/TikZit//triangle.tikz}}%
	\endpgfgraphicnamed
  & $T^{-1}$&$:$&$1\to 1$&%
	\beginpgfgraphicnamed{TikZit//triangleinv}
	\InputIfFileExists{TikZit//triangleinv.tikz}{}{\input{./figures/TikZit//triangleinv.tikz}}%
	\endpgfgraphicnamed
 \\\hline
\end{tabular}\caption{Generators of algebraic ZX-calculus, where $m,n\in \mathbb N$, $ a  \in \mathbb C$.} \label{qbzxgenerator} 
\end{center}
\end{table}
\FloatBarrier

 For simplicity, we make the following conventions: 
\[
	\beginpgfgraphicnamed{TikZit//spider0denote3}
	\InputIfFileExists{TikZit//spider0denote3.tikz}{}{\input{./figures/TikZit//spider0denote3.tikz}}%
	\endpgfgraphicnamed
\]
 \begin{equation} \label{andshortnotationeq}
	\beginpgfgraphicnamed{TikZit//andshortnote2}
	\InputIfFileExists{TikZit//andshortnote2.tikz}{}{\input{./figures/TikZit//andshortnote2.tikz}}%
	\endpgfgraphicnamed
 
      \end{equation}  
 where $\alpha \in \mathbb R, \tau \in \{ 0, \pi \}$.
 As a consequence, 
  \begin{equation}    \label{andshortnoterelat}
	\beginpgfgraphicnamed{TikZit//andshortnoterelation2}
	\InputIfFileExists{TikZit//andshortnoterelation2.tikz}{}{\input{./figures/TikZit//andshortnoterelation2.tikz}}%
	\endpgfgraphicnamed
 
      \end{equation}

 There is a standard interpretation $\left\llbracket \cdot \right\rrbracket$ for the ZX diagrams:
\[
\left\llbracket %
	\beginpgfgraphicnamed{TikZit//generalgreenspider}
	\InputIfFileExists{TikZit//generalgreenspider.tikz}{}{\input{./figures/TikZit//generalgreenspider.tikz}}%
	\endpgfgraphicnamed
 \right\rrbracket=\ket{0}^{\otimes m}\bra{0}^{\otimes n}+a\ket{1}^{\otimes m}\bra{1}^{\otimes n},
\left\llbracket %
	\beginpgfgraphicnamed{TikZit//redspider0p}
	\InputIfFileExists{TikZit//redspider0p.tikz}{}{\input{./figures/TikZit//redspider0p.tikz}}%
	\endpgfgraphicnamed
 \right\rrbracket=\sum_{\substack{0\leq i_1, \cdots, i_m,  j_1, \cdots, j_n\leq 1\\ i_1+\cdots+ i_m\equiv  j_1+\cdots +j_n(mod~ 2)}}\ket{i_1, \cdots, i_m}\bra{j_1, \cdots, j_n},
\]
\[
\left\llbracket%
	\beginpgfgraphicnamed{TikZit//newhadamard}
	\InputIfFileExists{TikZit//newhadamard.tikz}{}{\input{./figures/TikZit//newhadamard.tikz}}%
	\endpgfgraphicnamed
\right\rrbracket=\begin{pmatrix}
        1 & 1 \\
        1 & -1
 \end{pmatrix}, \quad \left\llbracket%
	\beginpgfgraphicnamed{TikZit//triangle}
	\InputIfFileExists{TikZit//triangle.tikz}{}{\input{./figures/TikZit//triangle.tikz}}%
	\endpgfgraphicnamed
\right\rrbracket=\begin{pmatrix}
        1 & 1 \\
        0 & 1
 \end{pmatrix}, \quad \quad
  \left\llbracket%
	\beginpgfgraphicnamed{TikZit//triangleinv}
	\InputIfFileExists{TikZit//triangleinv.tikz}{}{\input{./figures/TikZit//triangleinv.tikz}}%
	\endpgfgraphicnamed
\right\rrbracket=\begin{pmatrix}
        1 & -1 \\
        0 & 1
 \end{pmatrix}, \quad
\left\llbracket%
	\beginpgfgraphicnamed{TikZit//singleredpi}
	\InputIfFileExists{TikZit//singleredpi.tikz}{}{\input{./figures/TikZit//singleredpi.tikz}}%
	\endpgfgraphicnamed
\right\rrbracket=\begin{pmatrix}
        0 & 1 \\
        1 & 0
 \end{pmatrix}, \quad
\left\llbracket%
	\beginpgfgraphicnamed{TikZit//Id}
	\InputIfFileExists{TikZit//Id.tikz}{}{\input{./figures/TikZit//Id.tikz}}%
	\endpgfgraphicnamed
\right\rrbracket=\begin{pmatrix}
        1 & 0 \\
        0 & 1
 \end{pmatrix}, 
  \]

\[
 \left\llbracket%
	\beginpgfgraphicnamed{TikZit//swap}
	\InputIfFileExists{TikZit//swap.tikz}{}{\input{./figures/TikZit//swap.tikz}}%
	\endpgfgraphicnamed
\right\rrbracket=\begin{pmatrix}
        1 & 0 & 0 & 0 \\
        0 & 0 & 1 & 0 \\
        0 & 1 & 0 & 0 \\
        0 & 0 & 0 & 1 
 \end{pmatrix}, \quad
  \left\llbracket%
	\beginpgfgraphicnamed{TikZit//cap}
	\InputIfFileExists{TikZit//cap.tikz}{}{\input{./figures/TikZit//cap.tikz}}%
	\endpgfgraphicnamed
\right\rrbracket=\begin{pmatrix}
        1  \\
        0  \\
        0  \\
        1  \\
 \end{pmatrix}, \quad
   \left\llbracket%
	\beginpgfgraphicnamed{TikZit//cup}
	\InputIfFileExists{TikZit//cup.tikz}{}{\input{./figures/TikZit//cup.tikz}}%
	\endpgfgraphicnamed
\right\rrbracket=\begin{pmatrix}
        1 & 0 & 0 & 1 
         \end{pmatrix}, 
 \quad
  \left\llbracket%
	\beginpgfgraphicnamed{TikZit//emptysquare}
	\InputIfFileExists{TikZit//emptysquare.tikz}{}{\input{./figures/TikZit//emptysquare.tikz}}%
	\endpgfgraphicnamed
\right\rrbracket=1,  
   \]

\[  \llbracket D_1\otimes D_2  \rrbracket =  \llbracket D_1  \rrbracket \otimes  \llbracket  D_2  \rrbracket, \quad 
 \llbracket D_1\circ D_2  \rrbracket =  \llbracket D_1  \rrbracket \circ  \llbracket  D_2  \rrbracket,
  \]
where 
$$ a  \in \mathbb C, \quad \ket{0}= \begin{pmatrix}
        1  \\
        0  \\
 \end{pmatrix}, \quad 
 \bra{0}=\begin{pmatrix}
        1 & 0 
         \end{pmatrix},
 \quad  \ket{1}= \begin{pmatrix}
        0  \\
        1  \\
 \end{pmatrix}, \quad 
  \bra{1}=\begin{pmatrix}
     0 & 1 
         \end{pmatrix}.
 $$

\begin{remark}
For convenience, we often ignore the symbol of interpretation  $\left\llbracket  \right\rrbracket$ and equalise matrices and diagrams directly. 
  \end{remark}

 Next we present a set of rewriting rules for algebraic ZX-calculus \cite{qwangnormalformbit}.
  \begin{figure}[!h]
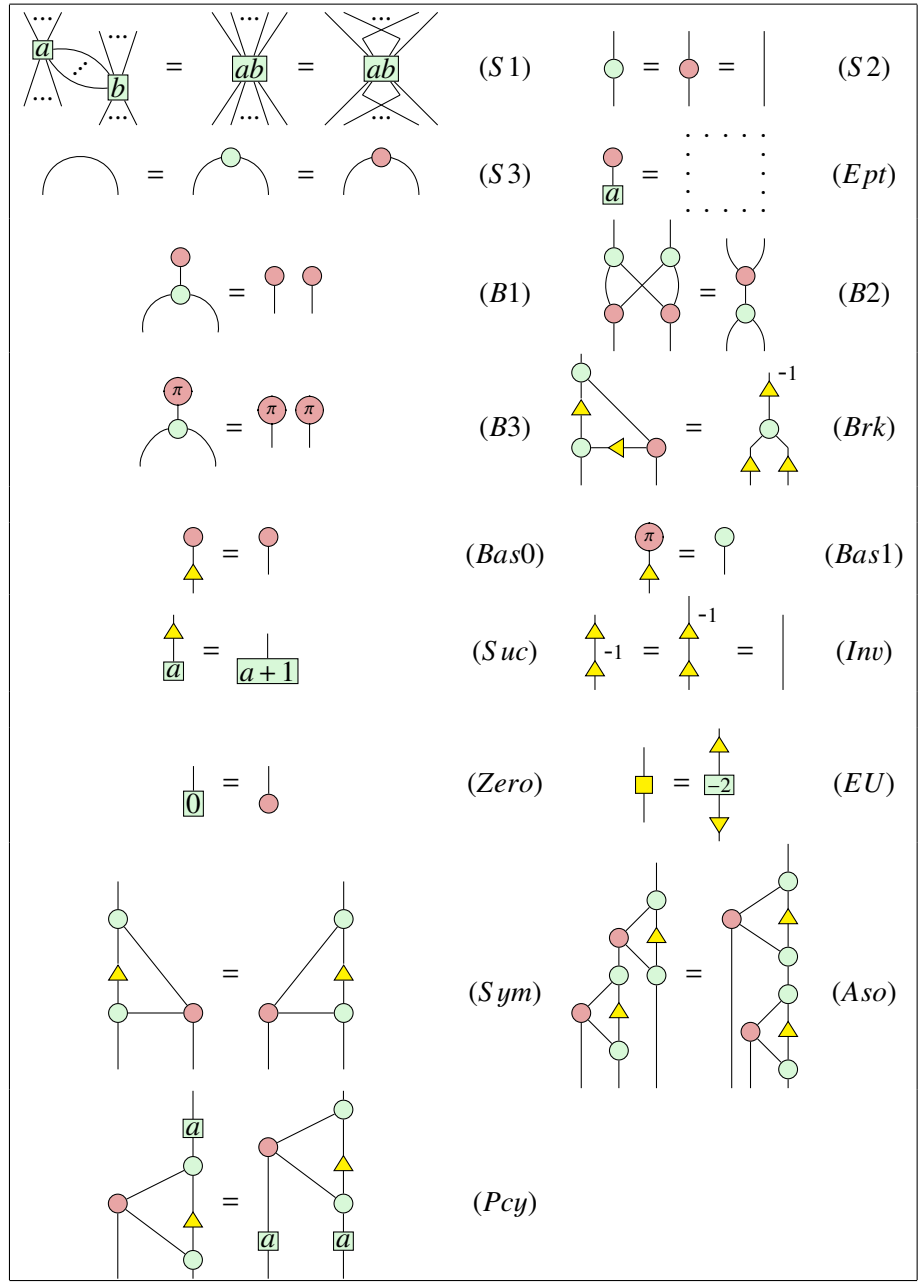

\begin{center}
\[
\quad \qquad\begin{array}{|cccc|}
\hline
	\beginpgfgraphicnamed{TikZit//generalgreenspiderfusesym}
	\InputIfFileExists{TikZit//generalgreenspiderfusesym.tikz}{}{\input{./figures/TikZit//generalgreenspiderfusesym.tikz}}%
	\endpgfgraphicnamed
&(S1) &%
	\beginpgfgraphicnamed{TikZit//s2new2}
	\InputIfFileExists{TikZit//s2new2.tikz}{}{\input{./figures/TikZit//s2new2.tikz}}%
	\endpgfgraphicnamed
 &(S2)\\
	\beginpgfgraphicnamed{TikZit//induced_compact_structure}
	\InputIfFileExists{TikZit//induced_compact_structure.tikz}{}{\input{./figures/TikZit//induced_compact_structure.tikz}}%
	\endpgfgraphicnamed
&(S3) & %
	\beginpgfgraphicnamed{TikZit//rdotaempty}
	\InputIfFileExists{TikZit//rdotaempty.tikz}{}{\input{./figures/TikZit//rdotaempty.tikz}}%
	\endpgfgraphicnamed
  &(Ept) \\
	\beginpgfgraphicnamed{TikZit//b1ring}
	\InputIfFileExists{TikZit//b1ring.tikz}{}{\input{./figures/TikZit//b1ring.tikz}}%
	\endpgfgraphicnamed
&(B1)  & %
	\beginpgfgraphicnamed{TikZit//b2ring}
	\InputIfFileExists{TikZit//b2ring.tikz}{}{\input{./figures/TikZit//b2ring.tikz}}%
	\endpgfgraphicnamed
&(B2)\\ 
  %
	\beginpgfgraphicnamed{TikZit//rpicopyns}
	\InputIfFileExists{TikZit//rpicopyns.tikz}{}{\input{./figures/TikZit//rpicopyns.tikz}}%
	\endpgfgraphicnamed
 &(B3)& %
	\beginpgfgraphicnamed{TikZit//anddflipns}
	\InputIfFileExists{TikZit//anddflipns.tikz}{}{\input{./figures/TikZit//anddflipns.tikz}}%
	\endpgfgraphicnamed
&(Brk) \\
 & &&\\
	\beginpgfgraphicnamed{TikZit//triangleocopy}
	\InputIfFileExists{TikZit//triangleocopy.tikz}{}{\input{./figures/TikZit//triangleocopy.tikz}}%
	\endpgfgraphicnamed
 &(Bas0) &%
	\beginpgfgraphicnamed{TikZit//trianglepicopyns}
	\InputIfFileExists{TikZit//trianglepicopyns.tikz}{}{\input{./figures/TikZit//trianglepicopyns.tikz}}%
	\endpgfgraphicnamed
&(Bas1)\\
	\beginpgfgraphicnamed{TikZit//plus1}
	\InputIfFileExists{TikZit//plus1.tikz}{}{\input{./figures/TikZit//plus1.tikz}}%
	\endpgfgraphicnamed
&(Suc)& %
	\beginpgfgraphicnamed{TikZit//triangleinvers}
	\InputIfFileExists{TikZit//triangleinvers.tikz}{}{\input{./figures/TikZit//triangleinvers.tikz}}%
	\endpgfgraphicnamed
  & (Inv) \\
   & &&\\
	\beginpgfgraphicnamed{TikZit//zerotoredns}
	\InputIfFileExists{TikZit//zerotoredns.tikz}{}{\input{./figures/TikZit//zerotoredns.tikz}}%
	\endpgfgraphicnamed
&(Zero)& %
	\beginpgfgraphicnamed{TikZit//eunoscalar}
	\InputIfFileExists{TikZit//eunoscalar.tikz}{}{\input{./figures/TikZit//eunoscalar.tikz}}%
	\endpgfgraphicnamed
&(EU) \\
	\beginpgfgraphicnamed{TikZit//lemma4}
	\InputIfFileExists{TikZit//lemma4.tikz}{}{\input{./figures/TikZit//lemma4.tikz}}%
	\endpgfgraphicnamed
&(Sym) &  %
	\beginpgfgraphicnamed{TikZit//associate}
	\InputIfFileExists{TikZit//associate.tikz}{}{\input{./figures/TikZit//associate.tikz}}%
	\endpgfgraphicnamed
 &(Aso)\\ 
	\beginpgfgraphicnamed{TikZit//TR1314combine2}
	\InputIfFileExists{TikZit//TR1314combine2.tikz}{}{\input{./figures/TikZit//TR1314combine2.tikz}}%
	\endpgfgraphicnamed
&(Pcy) &&\\ 
  		  		\hline
  		\end{array}\]
  	\end{center}
  	\caption{Algebraic rules, $a, b \in \mathbb C.$ The upside-down flipped versions of the rules are assumed to hold as well. }\label{figurealgebra}
  \end{figure}
 \FloatBarrier

\section{Elementary matrices in ZX diagrams }
 In this section, we show how to represent elementary matrices of size $2^m\times 2^m$  in algebraic ZX-calculus. Elementary matrices correspond to elementary operations on matrices: left multiplication  by an elementary matrix stands for elementary row operations, while right multiplication stands for elementary column operations.  
 
 There are three types of elementary matrices, the first type  performs the row (column) multiplication:
\[ R_{i\times(a)}=C_{i\times(a)}=
\begin{blockarray}{cccccl}
\begin{block}{(ccccc)l}
     1 & \cdots & 0 &\cdots & 0 &r_0\\
     \vdots    & \ddots & &&  \vdots&  \\
        0   & \cdots & a & \cdots& 0&r_i\\
       \vdots    & &  & \ddots&  \vdots &  \\
        0   & \cdots &0 & \cdots& 1&r_{2^m-1}\\
\end{block}
\end{blockarray}
 \] 
 The second type performs the row (column) addition:
 \[ R_{i\times(a)+ j}=C_{j\times(a)+ i}=
 \begin{blockarray}{cccccccl}
\begin{block}{(ccccccc)l}
     1 & \cdots & 0 &\cdots & 0&\cdots & 0 &r_0\\
     \vdots    & \ddots & &&&&  \vdots&  \\
        0   & \cdots & 1 & \cdots& 0& \cdots& 0&r_i\\
       \vdots    & &  & \ddots&  &&\vdots &  \\
        0   & \cdots &a & \cdots& 1& \cdots& 0&r_{j}\\
           \vdots    & &  & \ddots& && \vdots &  \\
               0   & \cdots &0 & \cdots& 0 &\cdots&1&r_{2^m-1}\\
\end{block}
\end{blockarray}
 \] 
 The third type performs the row (column) switching:
  \[ R_{i\leftrightarrow j}=C_{i\leftrightarrow j}=
 \begin{blockarray}{cccccccl}
\begin{block}{(ccccccc)l}
     1 & \cdots & 0 &\cdots & 0&\cdots & 0 &r_0\\
     \vdots    & \ddots & &&&&  \vdots&  \\
        0   & \cdots & 0 & \cdots& 1& \cdots& 0&r_i\\
       \vdots    & &  & \ddots&  &&\vdots &  \\
        0   & \cdots &1 & \cdots& 0& \cdots& 0&r_{j}\\
           \vdots    & &  & \ddots& && \vdots &  \\
               0   & \cdots &0 & \cdots& 0 &\cdots&1&r_{2^m-1}\\
\end{block}
\end{blockarray}
 \] 
Note that we count the rows from $0$ to $2^m-1$.
 
 Below we give the representation of the three kinds of elementary matrices in  Theorems \ref{rowitimesatm},  \ref{rowitimeaplusjtm}, and  \ref{swapmatrix}. Suppose $\{e_{k}~|~ 0\leq k \leq 2^{m}-1\}$ are the $2^m$-dimensional standard unit  column vectors (with entries all 0s except for one 1):
 \[
e_k=\begin{blockarray}{cl}
\begin{block}{(c)l}
     0 &r_0\\
      \vdots    &  \\
      1&r_k\\
     \vdots &  \\
       0&r_{2^m-1}\\
\end{block}
\end{blockarray}
 \]
 where $r_i$ denotes the i-th row, $0\leq i \leq 2^{m}-1,  m\geq 1$. Then
 \begin{equation}\label{qbitstovector}
$$\ket{a_{m-1}\cdots a_i \cdots a_0}=e_{\sum_{i=0}^{m-1}a_i2^i},$$
\end{equation}
where $a_i\in \{0, 1\}, 0\leq i \leq m-1$ \cite{Mermin}.

\begin{theorem}(Row multiplication)\label{rowitimesatm}
For any $2^m \times 2^m$ matrix, the $i$-th ($0 \leq i \leq 2^{m-1}$) row (column) multiplied by any number $a$ (including $0$) can be represented as
\[ 
R_{i\times(a)}
= %
	\beginpgfgraphicnamed{TikZit//rowitimesa}
	\InputIfFileExists{TikZit//rowitimesa.tikz}{}{\input{./figures/TikZit//rowitimesa.tikz}}%
	\endpgfgraphicnamed

\]
where $i=\sum_{k=0}^{m-1}a_k2^k, a_k \in \{0, 1\}, \bar{a}_k=a_k\oplus 1, \oplus$ is the modulo 2 addition.
\end{theorem}

\begin{remark}
In particular, if $m=1$, then the diagram becomes 
$$ %
	\beginpgfgraphicnamed{TikZit//rowitimesamwq1}
	\InputIfFileExists{TikZit//rowitimesamwq1.tikz}{}{\input{./figures/TikZit//rowitimesamwq1.tikz}}%
	\endpgfgraphicnamed
 $$
  \end{remark}


 


 \begin{example}\label{ex1}
To divide row 1 by 5, we have
$$R_{1\times(\frac{1}{5})}=  \begin{pmatrix}
        1 & 0 & 0 & 0 & 0 & 0 & 0 & 0\\
        0 & \frac{1}{5}&0&0&0&0 & 0&0 \\
           0 & 0 & 1 & 0 & 0 & 0 & 0 & 0\\ 
             0 & 0 & 0 & 1 & 0 & 0 & 0 & 0\\ 
   0 & 0 & 0 & 0 & 1 & 0 & 0 & 0\\ 
            0 & 0 & 0 & 0 & 0 & 1 & 0 & 0\\ 
              0 & 0 & 0 & 0 & 1 & 0 & 1 & 0\\ 
         0 & 0 & 0 & 0 & 0 & 0 & 0 & 1\\  
          \end{pmatrix} \quad = \quad %
	\beginpgfgraphicnamed{TikZit//row2time3rd0}
	\InputIfFileExists{TikZit//row2time3rd0.tikz}{}{\input{./figures/TikZit//row2time3rd0.tikz}}%
	\endpgfgraphicnamed
$$
Here $m=3, i=1=0\cdot 2^2+0\cdot 2^1+1\cdot 2^0,$ so $\bar{a}_2=1, \bar{a}_1=1, \bar{a}_0=0$.
\end{example}

\begin{theorem}(Row addition)\label{rowitimeaplusjtm}
Suppose  $i=a_{m-1}2^{m-1}+\cdots + a_{j_s}2^{j_s}+\cdots+a_{j_1}2^{j_1}+\cdots+ a_02^0, a_k \in \{0, 1\}, \bar{a}_k=a_k\oplus 1, \oplus$ is the modulo 2 addition, $j=a_{m-1}2^{m-1}+\cdots + \bar{a}_{j_s}2^{j_s}+\cdots+\bar{a}_{j_1}2^{j_1}+\cdots+ a_02^0$, $0\leq j_1< \cdots < j_s \leq m-1, 1\leq s \leq m$. In other words, $j_1, \cdots, j_s$ are exactly the bit positions where the numbers $i$ and $j$ differ. Then 
\[
R_{i\times(a)+ j}= %
	\beginpgfgraphicnamed{TikZit//rowitimeaplusj}
	\InputIfFileExists{TikZit//rowitimeaplusj.tikz}{}{\input{./figures/TikZit//rowitimeaplusj.tikz}}%
	\endpgfgraphicnamed

\]
\end{theorem}

\begin{remark}
In particular, if $m=1$, then the diagram becomes 
$$ %
	\beginpgfgraphicnamed{TikZit//rowitimesaplusjmwq1}
	\InputIfFileExists{TikZit//rowitimesaplusjmwq1.tikz}{}{\input{./figures/TikZit//rowitimesaplusjmwq1.tikz}}%
	\endpgfgraphicnamed
 $$
  \end{remark}

 



 \begin{example}\label{ex2}
  To add row 1, multiplied by 2, to row 3, we have
$$R_{1\times(2)+3}=  \begin{pmatrix}
        1 & 0 & 0 & 0 & 0 & 0 & 0 & 0\\
        0 & 1&0&0&0&0 & 0&0 \\
           0 & 0 & 1 & 0 & 0 & 0 & 0 & 0\\ 
             0 & 2 & 0 & 1 & 0 & 0 & 0 & 0\\ 
   0 & 0 & 0 & 0 & 1 & 0 & 0 & 0\\ 
            0 & 0 & 0 & 0 & 0 & 1 & 0 & 0\\ 
              0 & 0 & 0 & 0 & 1 & 0 & 1 & 0\\ 
         0 & 0 & 0 & 0 & 0 & 0 & 0 & 1\\  
          \end{pmatrix} \quad = \quad %
	\beginpgfgraphicnamed{TikZit//row1time2plus3}
	\InputIfFileExists{TikZit//row1time2plus3.tikz}{}{\input{./figures/TikZit//row1time2plus3.tikz}}%
	\endpgfgraphicnamed
$$
Here $m=3, i=1=0\cdot 2^2+0\cdot 2^1+1\cdot 2^0, j=3=0\cdot 2^2+1\cdot 2^1+1\cdot 2^0,$ so $\bar{a}_2=1, \bar{a}_1=1, \bar{a}_0=0, j_1=1, s=1$.
\end{example}

\begin{theorem}(Row switching)\label{swapmatrix}
Suppose  $m \geq 2, i=a_{m-1}2^{m-1}+\cdots + a_{j_s}2^{j_s}+\cdots+a_{j_k}2^{j_k}+ \cdots+a_{j_1}2^{j_1}+\cdots+ a_02^0, a_k \in \{0, 1\}, \bar{a}_k=a_k\oplus 1, \oplus$ is the modulo 2 addition, $j=a_{m-1}2^{m-1}+\cdots + \bar{a}_{j_s}2^{j_s}+\cdots+ \bar{a}_{j_k}2^{j_k}+\cdots+\bar{a}_{j_1}2^{j_1}+\cdots+ a_02^0$, $0\leq j_1< \cdots  < j_k < \cdots < j_s \leq m-1, 1\leq s \leq m$. In other words, $j_1, \cdots, j_s$ are exactly the bit positions where the numbers $i$ and $j$ differ.  Then 
\[
R_{i\leftrightarrow j}= %
	\beginpgfgraphicnamed{TikZit//rowitiswapj}
	\InputIfFileExists{TikZit//rowitiswapj.tikz}{}{\input{./figures/TikZit//rowitiswapj.tikz}}%
	\endpgfgraphicnamed

\]
\textbf{Note:} There are wires between the red spiders labelled $\overline{(a_{j_1} \oplus a_{j_x})}\pi$ and the green spider at wire $j_1$.
\end{theorem}

\begin{remark}
Note that for the special case
 when $s =1$, we have $a_{j_1}=a_{j_s}$, therefore 
 $\overline{(a_{j_1}+a_{j_s})}\pi=\pi$. Then
 \[
R_{i\leftrightarrow j}= %
	\beginpgfgraphicnamed{TikZit//rowitiswapjseq1}
	\InputIfFileExists{TikZit//rowitiswapjseq1.tikz}{}{\input{./figures/TikZit//rowitiswapjseq1.tikz}}%
	\endpgfgraphicnamed

\]
Hence there are always exactly $m-1$ inputs to the AND gate.
Also for the  special case
 when $m=s =2$,  we have $a_{j_k}=a_{j_s}, a_0=\bar{a}_1$, therefore  $\overline{(a_{1}+a_{0})}\pi=0$. Then
  \[
R_{i\leftrightarrow j}= %
	\beginpgfgraphicnamed{TikZit//rowitiswapjseq2}
	\InputIfFileExists{TikZit//rowitiswapjseq2.tikz}{}{\input{./figures/TikZit//rowitiswapjseq2.tikz}}%
	\endpgfgraphicnamed

\]
If $m=1$, then $s=1$ since $i\neq j$, therefore 
$$R_{i\leftrightarrow j}= %
	\beginpgfgraphicnamed{TikZit//rowitiswapjmeq1}
	\InputIfFileExists{TikZit//rowitiswapjmeq1.tikz}{}{\input{./figures/TikZit//rowitiswapjmeq1.tikz}}%
	\endpgfgraphicnamed

 $$
\end{remark}



 \begin{example}\label{ex3}
To swap rows 1 and 7, we have
\begin{center}
$R_{1\leftrightarrow 7}=  \begin{pmatrix}
        1 & 0 & 0 & 0 & 0 & 0 & 0 & 0\\
        0 & 0 & 0 & 0 & 0 & 0 & 0 & 1\\  
           0 & 0 & 1 & 0 & 0 & 0 & 0 & 0\\ 
             0 & 0 & 0 & 1 & 0 & 0 & 0 & 0\\ 
   0 & 0 & 0 & 0 & 1 & 0 & 0 & 0\\ 
            0 & 0 & 0 & 0 & 0 & 1 & 0 & 0\\ 
              0 & 0 & 0 & 0 & 0 & 0 & 1 & 0\\ 
                 0 & 1&0&0&0&0 & 0&0 \\
         \end{pmatrix} \quad = \quad %
	\beginpgfgraphicnamed{TikZit//switch28v2}
	\InputIfFileExists{TikZit//switch28v2.tikz}{}{\input{./figures/TikZit//switch28v2.tikz}}%
	\endpgfgraphicnamed
$
\end{center}
         Here $m=3, i=1=0\cdot 2^2+0\cdot 2^1+1\cdot 2^0, j=7=1\cdot 2^2+1\cdot 2^1+1\cdot 2^0,$ so $\bar{a}_2=1, \bar{a}_1=1, \bar{a}_0=0, j_1=1, j_2=2, s=2$.
\end{example}

 

  \section{Properties of diagrammatic elementary matrices}\label{properties}  
  In this section, we prove by diagrams some properties of elementary matrices on their inverses and transpose. The proofs are given in the appendix. 
  

  \begin{proposition}\label{multiplicationinv} 
 Suppose $a \neq 0, a \in \mathbb{C}$. Then $ R_{i\times(a)}^{-1}=R_{i\times(\frac{1}{a})}, i.e.,  R_{i\times(a)}R_{i\times(\frac{1}{a})}=R_{i\times(\frac{1}{a})}R_{i\times(a)}=I$.
   \end{proposition}

 \begin{proposition}\label{additioninv}  
  $ R_{i\times(a)+ j}^{-1}=R_{i\times(-a)+ j}, i.e.,  R_{i\times(-a)+ j}R_{i\times(a)+ j}=R_{i\times(a)+ j}R_{i\times(-a)+ j}=I$.
   \end{proposition}

  \begin{proposition}\label{sawpinv}   
  $ R_{i\leftrightarrow j}^{-1}=R_{i\leftrightarrow j}, i.e.,  R_{i\leftrightarrow j}R_{i\leftrightarrow j}=I$.
   \end{proposition}   
  
  \begin{proposition}\label{multiplytranp}    
  $ R_{i\times(a)}^T=R_{i\times(a)}$.
   \end{proposition}   
   
   \begin{proposition}\label{additiontranp} 
  $ R_{i\times(a)+ j}^T=R_{j\times(a)+ i}$.
   \end{proposition}    
    
   \begin{proposition}\label{sawptranp}  
  $ R_{i\leftrightarrow j}^T=R_{i\leftrightarrow j}$.
   \end{proposition}

   \begin{proposition} 
  $ R_{i\times(a)}R_{i\times(b)}=R_{i\times(ab)}$.
   \end{proposition}  
This follows directly from  Lemma \ref{pigatecopy}, Lemma \ref{generalbialgebra} and the rule (S1).
   \begin{proposition} 
  $ R_{i\times(a)+ j}R_{i\times(b)+ j}=R_{i\times(a+b)+ j}$.
   \end{proposition}  
  This follows directly from  Lemma \ref{pigatecopy} and Lemma \ref{propadprimecro}.
  
\section{Represent arbitrary matrices by string diagrams}    
In this section we show how to represent an arbitrary matrix $A$ of size $2^m\times 2^n$ as ZX diagrams. An implementation of this method is available on GitHub\footnote{https://github.com/y-richie-y/qpl-represent-matrix}. The idea is to use elementary transformations to turn  $A$  into a simple matrix which can be easily represented in ZX, then we apply the inverse operations diagrammatically to get the diagram for $A$.

As is well known in linear algebra, any matrix $A$ can be turned into a reduced row echelon form by means of a finite sequence of elementary row operations. If we further allow elementary column operations to be used, then $A$ can be transformed to a standard form
$
C=  \begin{pmatrix}
        E_r & O\\
         O & O
          \end{pmatrix}_{2^m\times 2^n}
$
where $r$ is the rank of $A$, and $E_r$ is an identity matrix of order $r$. Below we show that  $C$ can be represented as a ZX diagram in the proof of the following theorem which is proved in the appendix.

  \begin{theorem}\label{anymatpresen} 
Any matrix $A$ of size $2^m\times 2^n$ can be represented by a ZX diagram.

   \end{theorem}  
  \begin{remark}
In practice, it is not necessary to go from the beginning to the last step (standard form). You can stop at anywhere you know how to represent the corresponding matrix in ZX diagrams.
  \end{remark} 
   
 \begin{example}\label{2by2matrix}
Given an arbitrary $2\times 2$ matrix $A$, let
$A=\begin{pmatrix}
       a & b\\
       c & d
          \end{pmatrix}.
          $
If $A$ is a zero matrix, then $A=  %
	\beginpgfgraphicnamed{TikZit//0matrix2by2}
	\InputIfFileExists{TikZit//0matrix2by2.tikz}{}{\input{./figures/TikZit//0matrix2by2.tikz}}%
	\endpgfgraphicnamed
$. Otherwise, we assume $a\neq 0$ (if $a=0$ then we can swap the location of $a$ with that of a non-zero element via elementary transformations).  Then
 \[  \begin{pmatrix}
         a & b\\
       c & d
                     \end{pmatrix} 
                     \overset{R_{0\times(\frac{1}{a})}}{\longrightarrow}
                     \begin{pmatrix}
        1 & \frac{b}{a}\\
       c & d
          \end{pmatrix}
         \overset{R_{0\times(-c)+ 1}}{\longrightarrow}
         \begin{pmatrix}
          1 & \frac{b}{a}\\
       0 & d- \frac{bc}{a}
             \end{pmatrix}    \overset{C_{0\times(-\frac{b}{a})+ 1}}{\longrightarrow}
             \begin{pmatrix}
        1 &0\\
       0 & d- \frac{bc}{a}
             \end{pmatrix}
             = %
	\beginpgfgraphicnamed{TikZit//simplified2by2mt}
	\InputIfFileExists{TikZit//simplified2by2mt.tikz}{}{\input{./figures/TikZit//simplified2by2mt.tikz}}%
	\endpgfgraphicnamed

\]
where $k=d- \frac{bc}{a}$.
Reverse the procedure, considering $i=0\times 2^0, a_0=0, j=1\times 2^0$, then we obtain the diagram for  the matrix $X$:
\[
A=%
	\beginpgfgraphicnamed{TikZit//2by2dm}
	\InputIfFileExists{TikZit//2by2dm.tikz}{}{\input{./figures/TikZit//2by2dm.tikz}}%
	\endpgfgraphicnamed

\]
Based on the representation of $A$ and the method from \cite{shaikhHowSumExponentiate2022}, the controlled-$A$ matrix can be given as follows:
\[
	\beginpgfgraphicnamed{controlledmatrix}
	\InputIfFileExists{controlledmatrix.tikz}{}{\input{./figures/controlledmatrix.tikz}}%
	\endpgfgraphicnamed

\]

\end{example} 

\section{Representation of matchgates}\label{match}
 In this section, we consider a very interesting class of two–qubit gates called matchgates (Jozsa style) \cite{jozsa2008matchgates}, which have been used for efficient computing on a classical computer \cite{jozsa2008matchgates} and universal quantum computation \cite{Jozsa_2009}. We want to use the representation method from the previous section to represent the matchgate in algebraic ZX diagrams.

      A matchgate $G(A,B)$ has the form
        \[
 G(A,B)= \begin{pmatrix}
    p & 0 & 0 & q     \\
    0 & w & x & 0   \\
    0 & y & z & 0   \\
 r & 0 & 0 & s \\
  \end{pmatrix}, \quad 
  A=   \begin{pmatrix}
    p  & q     \\
 r & s \\
  \end{pmatrix}, \quad 
  B=   \begin{pmatrix}
w & x \\
 y & z \\
  \end{pmatrix}, 
  \]
where $A, B$ are both in the special unitary group $SU(2)$.

 Then $G(A,B)$ can be represented in ZX as follows: 
      \[
	\beginpgfgraphicnamed{elementarysimplydoubled3}
	\InputIfFileExists{elementarysimplydoubled3.tikz}{}{\input{./figures/elementarysimplydoubled3.tikz}}%
	\endpgfgraphicnamed
 \quad =  \quad%
	\beginpgfgraphicnamed{elementarysimplydoubledsim}
	\InputIfFileExists{elementarysimplydoubledsim.tikz}{}{\input{./figures/elementarysimplydoubledsim.tikz}}%
	\endpgfgraphicnamed

      \] 
The details on how to get this representation of matchgates by elementary matrix diagrams can be found in the appendix. We note that  the matchgate in quantum circuit form and ZW diagrams has been shown by Niel de Beaudrap and Miriam Backens respectively at ZX discord, which we attached in the appendix. Our representation is more compact in the sense that we only have two controlled operations (the bottom part with triangle and $\sqrt{\frac{ry}{pw}}$ box and the top part with triangle and $\sqrt{\frac{xq}{pw}}$), while the circuit form (which has 2 controlled unitaries) has 6 controlled operations in general case (each controlled unitary is composed of three controlled rotations due to the Euler decomposition in general). The ZW form needs 4 W nodes with 1 input and 3 outputs, while our representation just needs 2 W nodes with 1 input and 2 outputs (showed in the appendix).

\section{Software Implementation}
We implemented Lemmas \ref{rowitimesatm}, \ref{rowitimeaplusjtm}, and \ref{swapmatrix} into \texttt{discopy} \cite{de2020discopy,toumi2022discopy} and numerically verified their correctness using tensor network contraction. Then we combine these lemmas with Gaussian elimination to automatically synthesise $2^m$ by $2^m$ matrices. We chose to extend the ZX module of \texttt{discopy} rather than \texttt{pyzx} because its representation for nodes do not make the distinction between input and output edges, and so cannot natively represent the non-symmetric triangle node. Here are the diagrams for Examples \ref{ex1}, \ref{ex2}, and \ref{ex3} produced by \texttt{discopy}.
\begin{center}
\includegraphics[height=5cm]{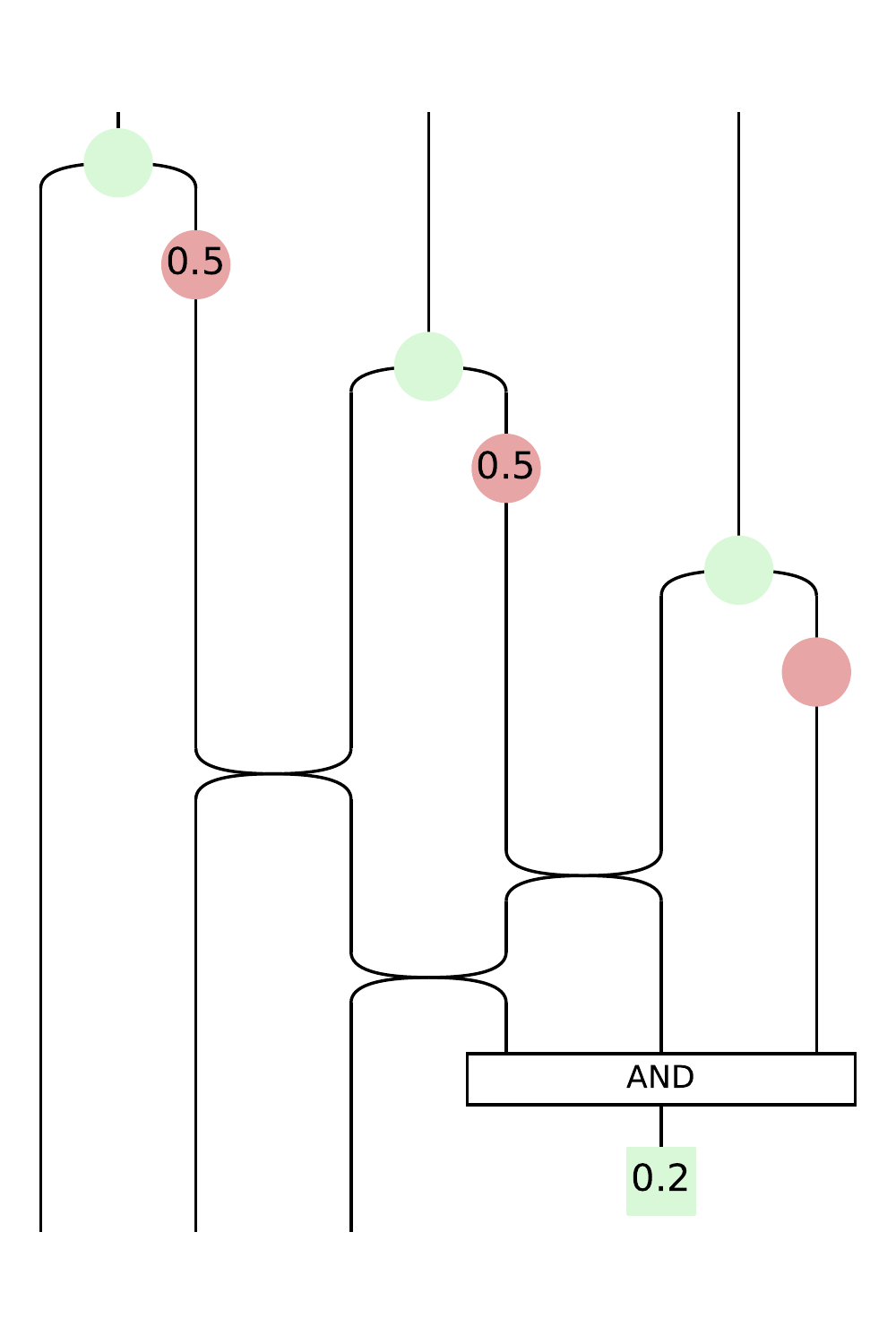} \hspace{1cm}
\includegraphics[height=5cm]{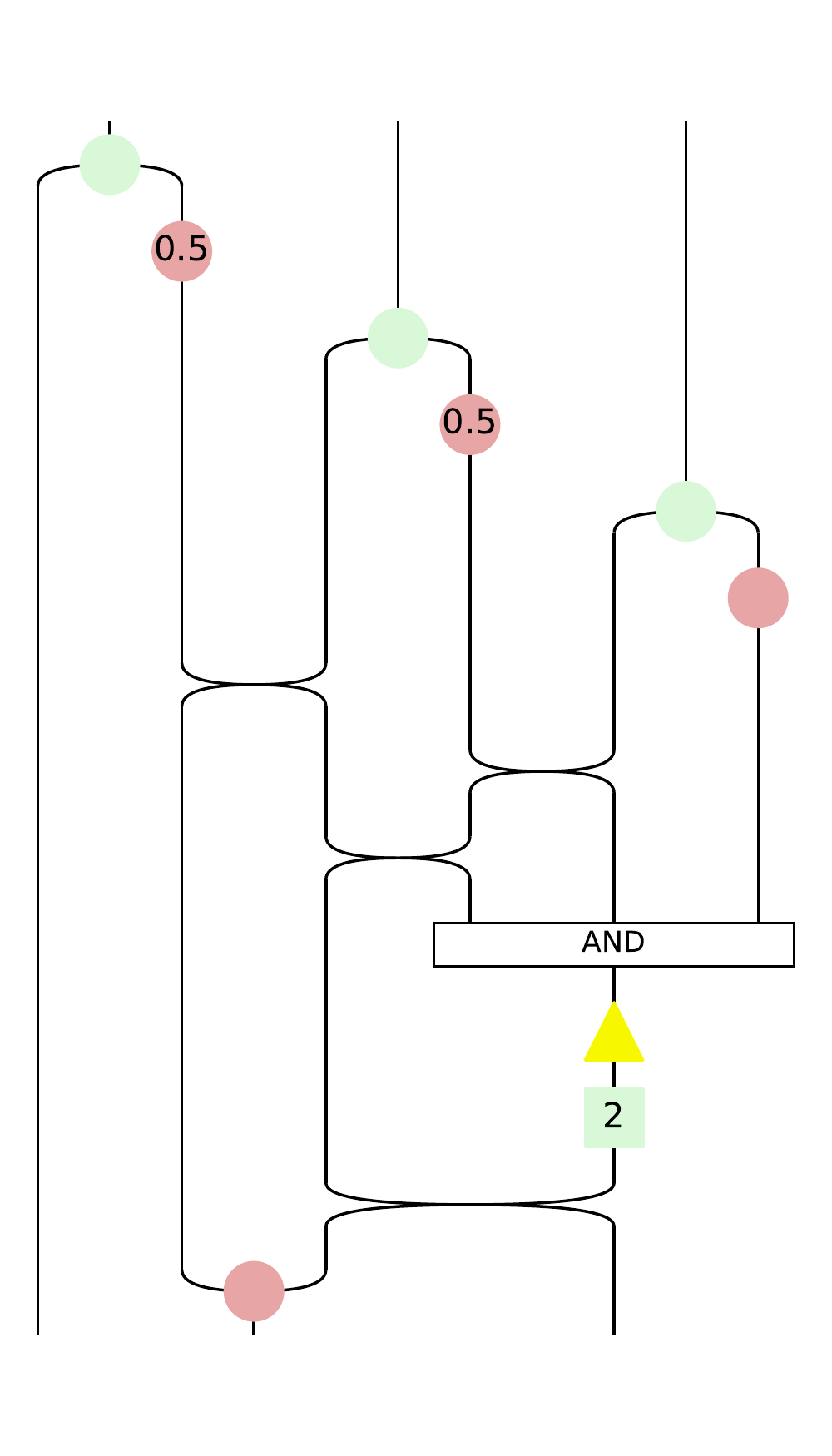} \hspace{1cm}
\includegraphics[height=5cm]{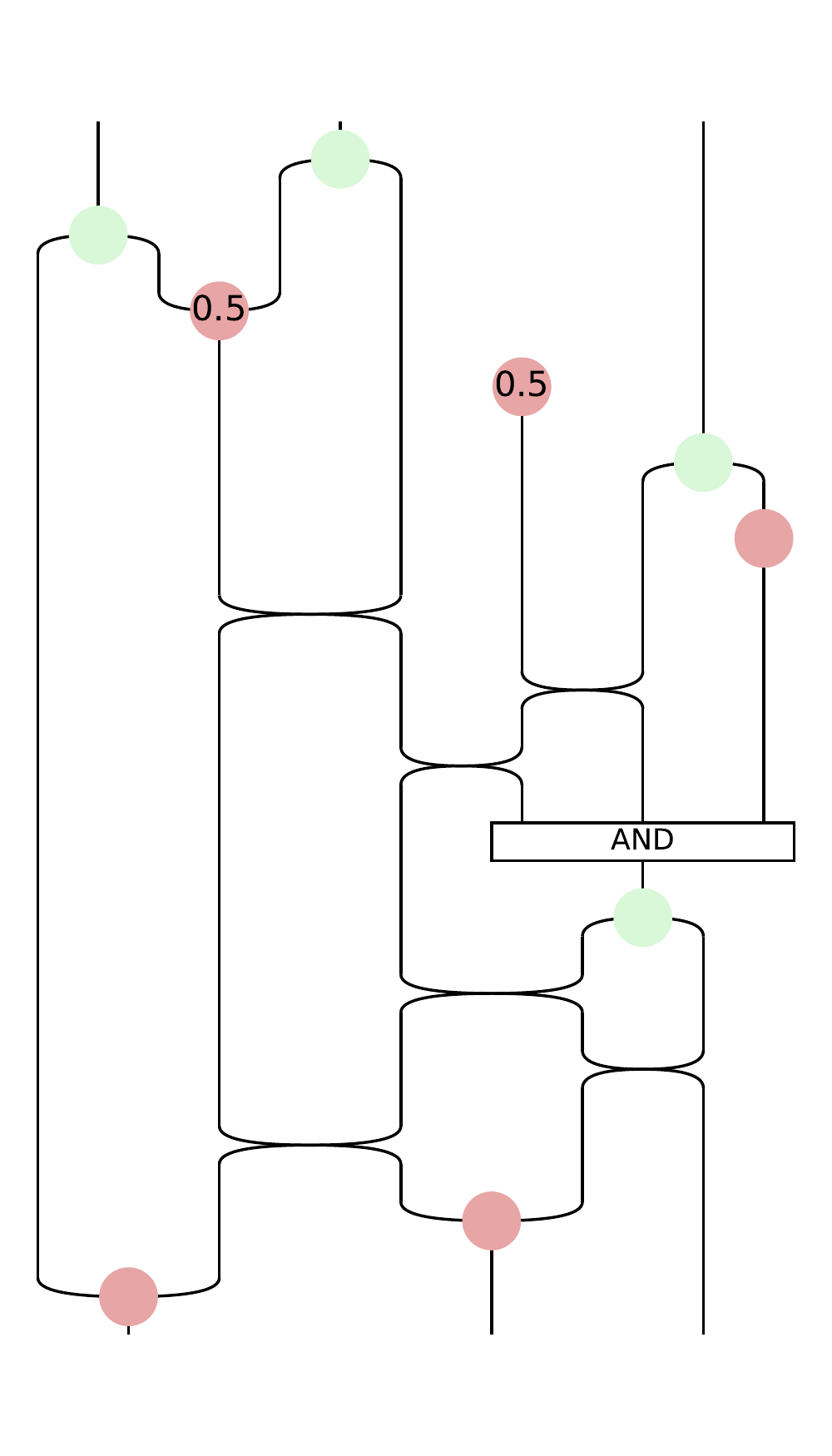}
\end{center}

\noindent
Note that \texttt{discopy}'s phases range from $[0, 1)$ rather than $[0, 2\pi)$.
The implementation can be found at \url{https://github.com/y-richie-y/qpl-represent-matrix}.

\section{Further work}  
In this paper, we first show how to represent elementary matrices of size  $2^m\times 2^m$ in ZX diagrams. Based on that, we depict arbitrary matrix of size $2^m\times 2^n$ via algebraic ZX-calculus. Also we implemented this representation in \texttt{discopy}. Furthermore, we show how this representation method could be used for representing matchgates.

Next, we would like to represent some useful matrix technologies like singular value decomposition (SVD), QR decomposition and lower–upper (LU) decomposition using algebraic ZX-calculus, which we hope could pave the way towards visualising  important matrix technologies deployed in machine learning. We would also like to combine our matrix representation tool written in \texttt{discopy} \cite{de2020discopy} with rewrite software such as \texttt{pyzx} \cite{kissinger2019pyzx}.

Although there is a normal form for arbitrary finite dimensional vectors \cite{qwangqufinite}, which means any finite matrix can be represented in higher dimensional ZX diagrams via map-state duality \cite{Coeckebk}, we still look for a simple representation for all kinds of elementary matrices of arbitrary size (some kinds of elementary matrices of size $d^m\times d^m$ for any positive integer $d$ have been given in \cite{qwangqufinite}), so that  any finite matrix can be represented in ZX diagrams in a similar way as shown in this paper.
  \section*{Acknowledgements}  
 The authors thank Konstantinos Meichanetzidis and Bob Coecke for insightful  comments. The first author would like to acknowledge the grant FQXi-RFP-CPW-2018 as well.    
 

\printbibliography

\section*{Appendix}
In this appendix, we include all the proofs of lemmas, theorems and propositions.

 \begin{proof}[Proof of \autoref{rowitimesatm}]
Note that $a_k\oplus \bar{a}_k=1$. Then we have
\[ 
	\beginpgfgraphicnamed{TikZit//rowitimesaprf}
	\InputIfFileExists{TikZit//rowitimesaprf.tikz}{}{\input{./figures/TikZit//rowitimesaprf.tikz}}%
	\endpgfgraphicnamed
=a e_{\sum_{k=0}^{m-1}a_k2^k}
\]
For any other inputting state $c_{m-1}\pi, \cdots c_k\pi, \cdots, c_0\pi$, there must be a $j$ such that $a_j\neq c_j$, i.e., $c_j=\bar{a}_j$, therefore one get 
\[ 
	\beginpgfgraphicnamed{TikZit//rowitimesaprf2}
	\InputIfFileExists{TikZit//rowitimesaprf2.tikz}{}{\input{./figures/TikZit//rowitimesaprf2.tikz}}%
	\endpgfgraphicnamed
= e_{\sum_{k=0}^{m-1}c_k2^k}
\]
 \end{proof}
 
 \begin{proof}[Proof of \autoref{rowitimeaplusjtm}]
 
The $i$-th column of $R_{i\times(a)+ j}$ should be $e_i+ae_j$, and the other columns are just of the form $e_k$.
First for the $i$-th column we have
\[ 
	\beginpgfgraphicnamed{TikZit//rowitimeaplusjprf}
	\InputIfFileExists{TikZit//rowitimeaplusjprf.tikz}{}{\input{./figures/TikZit//rowitimeaplusjprf.tikz}}%
	\endpgfgraphicnamed
\]
 \[=\ket{\underset{m-1}{a_{m-1}}\cdots \underset{j_s}{a_{j_s}}\cdots \underset{j_1}{a_{j_1}}\cdots \underset{0}{a_0}}+a\ket{\underset{m-1}{a_{m-1}}\cdots \underset{j_s}{\bar{a}_{j_s}}\cdots \underset{j_1}{\bar{a}_{j_1}}\cdots \underset{0}{a_0}}=e_i+ae_j
\]
For any other inputting state $c_{m-1}\pi, \cdots c_k\pi, \cdots, c_0\pi$, there must be a $j$ such that $a_j\neq c_j$, i.e., $c_j=\bar{a}_j$, therefore one get 
\[ 
	\beginpgfgraphicnamed{TikZit//rowitimeaplusjprf2}
	\InputIfFileExists{TikZit//rowitimeaplusjprf2.tikz}{}{\input{./figures/TikZit//rowitimeaplusjprf2.tikz}}%
	\endpgfgraphicnamed
= e_{\sum_{k=0}^{m-1}c_k2^k}
\]
  \end{proof}
  
  \begin{proof}[Proof of \autoref{swapmatrix}] 
The state $\ket{a_{m-1}\cdots a_{j_s}\cdots a_{j_k}\cdots a_{j_1}\cdots a_{0}}$ will be sent to $\ket{a_{m-1}\cdots \bar{a}_{j_s}\cdots \bar{a}_{j_k}\cdots \bar{a}_{j_1}\cdots a_{0}}$  by $R_{i\leftrightarrow j}$ and vice versa, while  the other states are remain unchanged under the action of $R_{i\leftrightarrow j}$.
First for the $i$-th column we have
\[ 
	\beginpgfgraphicnamed{TikZit//rowitiswapjprf}
	\InputIfFileExists{TikZit//rowitiswapjprf.tikz}{}{\input{./figures/TikZit//rowitiswapjprf.tikz}}%
	\endpgfgraphicnamed
\]
Similarly, we have
\[ 
	\beginpgfgraphicnamed{TikZit//rowitiswapjprf2}
	\InputIfFileExists{TikZit//rowitiswapjprf2.tikz}{}{\input{./figures/TikZit//rowitiswapjprf2.tikz}}%
	\endpgfgraphicnamed
\]
 
For any other input state $c_{m-1}\pi, \cdots c_k\pi, \cdots, c_0\pi$ (corresponding to $u$), if $u$ is different from $i$ and $j$ in some place that beyond the set $\{ j_1, \cdots, j_s\}$, then  there must be a $j$ such that $a_j\neq c_j$, i.e., $c_j=\bar{a}_j$, therefore one get 
\[ 
	\beginpgfgraphicnamed{TikZit//rowitiswapjprf3}
	\InputIfFileExists{TikZit//rowitiswapjprf3.tikz}{}{\input{./figures/TikZit//rowitiswapjprf3.tikz}}%
	\endpgfgraphicnamed

\]
If $u$ is different from $i$ and $j$ only in some places that belong to the set $\{ j_1, \cdots, j_s\}$, say $j_s$ and $j_k$ (since $i$ and $j$ are different in these places, the number of places where $u$ is different from $i$ and $j$  must be at least $2$). If $c_{j_k}=a_{j_k}$, then $c_{j_s}\neq a_{j_s}$ (otherwise $u$ won't be different from $i$ at places $j_s$ and $j_k$), i.e., $c_{j_s}= \bar{a}_{j_s}$. Therefore, $c_{j_k}\oplus c_{j_s}=a_{j_k}\oplus \bar{a}_{j_s}\neq a_{j_k}\oplus a_{j_s}$. Similarly,  If $c_{j_k}=\bar{a}_{j_k}$, then we still have $c_{j_k}\oplus c_{j_s}\neq a_{j_k}\oplus a_{j_s}$.  Now we claim that it is impossible that both  $c_{j_k}\oplus c_{j_1}=a_{j_k}\oplus a_{j_1}$ and $c_{j_s}\oplus c_{j_1}=a_{j_s}\oplus a_{j_1}$ hold. Otherwise, we could sum up (modulo 2) these two equalities and then get   $c_{j_k}\oplus c_{j_s}= a_{j_k}\oplus a_{j_s}$ which is a contradiction. Hence either  $c_{j_k}\oplus c_{j_1}=\overline{a_{j_k}\oplus a_{j_1}}$ or  $c_{j_s}\oplus c_{j_1}=\overline{a_{j_s}\oplus a_{j_1}}$. Assume $c_{j_s}\oplus c_{j_1}=\overline{a_{j_s}\oplus a_{j_1}}$, then we have
\[ 
	\beginpgfgraphicnamed{TikZit//rowitiswapjprf4}
	\InputIfFileExists{TikZit//rowitiswapjprf4.tikz}{}{\input{./figures/TikZit//rowitiswapjprf4.tikz}}%
	\endpgfgraphicnamed

\]
\end{proof}

\section*{Lemmas for \autoref{properties}}

  \begin{lemma}\cite{qwangnormalformbit}\label{pigatecopy}
	\beginpgfgraphicnamed{TikZit//pimultiplecp}
	\InputIfFileExists{TikZit//pimultiplecp.tikz}{}{\input{./figures/TikZit//pimultiplecp.tikz}}%
	\endpgfgraphicnamed
 
   \end{lemma} 
   
    \begin{lemma}\cite{duncan_graph_2009}\label{generalbialgebralm}
 $$ %
	\beginpgfgraphicnamed{TikZit//strongcomplementaryn}
	\InputIfFileExists{TikZit//strongcomplementaryn.tikz}{}{\input{./figures/TikZit//strongcomplementaryn.tikz}}%
	\endpgfgraphicnamed
$$
 \end{lemma}   
     \begin{lemma}\cite{qwangnormalformbit}\label{generalbialgebra}
	\beginpgfgraphicnamed{TikZit//generalBiA}
	\InputIfFileExists{TikZit//generalBiA.tikz}{}{\input{./figures/TikZit//generalBiA.tikz}}%
	\endpgfgraphicnamed
 
 \end{lemma}
 
  \begin{lemma}\cite{bobanthonywang}\label{generalbialgebra}
   For any $k\geq 0$, we have   
$$%
	\beginpgfgraphicnamed{TikZit//generalBiAvariant}
	\InputIfFileExists{TikZit//generalBiAvariant.tikz}{}{\input{./figures/TikZit//generalBiAvariant.tikz}}%
	\endpgfgraphicnamed
 $$  
    \end{lemma}  
   
    \begin{lemma}\label{rowitiswapjequivlm}
\[
R_{i\leftrightarrow j}= %
	\beginpgfgraphicnamed{TikZit//rowitiswapjequiv}
	\InputIfFileExists{TikZit//rowitiswapjequiv.tikz}{}{\input{./figures/TikZit//rowitiswapjequiv.tikz}}%
	\endpgfgraphicnamed

\]
where $b_s= \overline{(a_{j_1}\oplus a_{j_s})}, \cdots, b_k= \overline{(a_{j_1}\oplus a_{j_k})}.$
\end{lemma}
\begin{proof}
\[ 
	\beginpgfgraphicnamed{TikZit//rowitiswapjequivprf}
	\InputIfFileExists{TikZit//rowitiswapjequivprf.tikz}{}{\input{./figures/TikZit//rowitiswapjequivprf.tikz}}%
	\endpgfgraphicnamed

\]
\[ 
	\beginpgfgraphicnamed{TikZit//rowitiswapjequivprf2}
	\InputIfFileExists{TikZit//rowitiswapjequivprf2.tikz}{}{\input{./figures/TikZit//rowitiswapjequivprf2.tikz}}%
	\endpgfgraphicnamed

\]
where the last equality is obtained using the same method as previous steps.
\end{proof}

    \begin{lemma}\cite{qwangnormalformbit}\label{propadprimecro}
 $$ %
	\beginpgfgraphicnamed{TikZit//propaddprimecro}
	\InputIfFileExists{TikZit//propaddprimecro.tikz}{}{\input{./figures/TikZit//propaddprimecro.tikz}}%
	\endpgfgraphicnamed
$$
 \end{lemma} 
 
 \begin{lemma}\cite{qwangnormalformbit}\label{trianglehopflip}
 \begin{equation*}\label{trianglehopflipeq}
	\beginpgfgraphicnamed{TikZit//trianglehopfflip2}
	\InputIfFileExists{TikZit//trianglehopfflip2.tikz}{}{\input{./figures/TikZit//trianglehopfflip2.tikz}}%
	\endpgfgraphicnamed
 
   \end{equation*} 
    \end{lemma}

  \begin{proof}[Proof of \autoref{multiplicationinv}] 
 $R_{i\times(\frac{1}{a})}R_{i\times(a)}
=$
   \[ 
	\beginpgfgraphicnamed{TikZit//rowitimesainverseprf}
	\InputIfFileExists{TikZit//rowitimesainverseprf.tikz}{}{\input{./figures/TikZit//rowitimesainverseprf.tikz}}%
	\endpgfgraphicnamed

\]
$ R_{i\times(a)}R_{i\times(\frac{1}{a})}=I$ can be proved similarly.  
  \end{proof}

    \begin{proof}[Proof of \autoref{additioninv}] 
 $R_{i\times(-a)+ j}R_{i\times(a)+ j}
=$
   \[ 
	\beginpgfgraphicnamed{TikZit//rowitimeaplusjinverseprf}
	\InputIfFileExists{TikZit//rowitimeaplusjinverseprf.tikz}{}{\input{./figures/TikZit//rowitimeaplusjinverseprf.tikz}}%
	\endpgfgraphicnamed

\]
 \[ 
	\beginpgfgraphicnamed{TikZit//rowitimeaplusjinverseprfii}
	\InputIfFileExists{TikZit//rowitimeaplusjinverseprfii.tikz}{}{\input{./figures/TikZit//rowitimeaplusjinverseprfii.tikz}}%
	\endpgfgraphicnamed

\]
$ R_{i\times(a)+ j}R_{i\times(-a)+ j}=I$ can be proved similarly.  
  \end{proof}    

    \begin{proof}[Proof of \autoref{sawpinv}] 
 $R_{i\leftrightarrow j}R_{i\leftrightarrow j}
=$
   \[ 
	\beginpgfgraphicnamed{TikZit//rowswapinverseprf}
	\InputIfFileExists{TikZit//rowswapinverseprf.tikz}{}{\input{./figures/TikZit//rowswapinverseprf.tikz}}%
	\endpgfgraphicnamed

\]
\[ 
	\beginpgfgraphicnamed{TikZit//rowitiswapjequivprfii}
	\InputIfFileExists{TikZit//rowitiswapjequivprfii.tikz}{}{\input{./figures/TikZit//rowitiswapjequivprfii.tikz}}%
	\endpgfgraphicnamed

\]
  \end{proof}  

 \begin{proof}[Proof of \autoref{multiplytranp}] 
 $R_{i\times(a)}^T=$
   \[ 
	\beginpgfgraphicnamed{TikZit//rowtimestransposeprf}
	\InputIfFileExists{TikZit//rowtimestransposeprf.tikz}{}{\input{./figures/TikZit//rowtimestransposeprf.tikz}}%
	\endpgfgraphicnamed
\]
 \[
 =R_{i\times(a)}
\]
  \end{proof}  

 \begin{proof}[Proof of \autoref{additiontranp}] 
    Note that  $i=a_{m-1}2^{m-1}+\cdots + a_{j_s}2^{j_s}+\cdots+a_{j_1}2^{j_1}+\cdots+ a_02^0, a_k \in \{0, 1\}, \bar{a}_k=a_k\oplus 1, \oplus$ is the modulo 2 addition, $j=a_{m-1}2^{m-1}+\cdots + \bar{a}_{j_s}2^{j_s}+\cdots+\bar{a}_{j_1}2^{j_1}+\cdots+ a_02^0$, $0\leq j_1< \cdots < j_s \leq m-1, 1\leq s \leq m$, so $j$ is different from $i$ exactly in the $j_1, \cdots, j_s$ places in their binary expansions. 
 Then $R_{i\times(a)+ j}^T=$
   \[ 
	\beginpgfgraphicnamed{TikZit//rowtimesplustransposeprf}
	\InputIfFileExists{TikZit//rowtimesplustransposeprf.tikz}{}{\input{./figures/TikZit//rowtimesplustransposeprf.tikz}}%
	\endpgfgraphicnamed
\]
    \[ 
	\beginpgfgraphicnamed{TikZit//rowtimesplustransposeprf2}
	\InputIfFileExists{TikZit//rowtimesplustransposeprf2.tikz}{}{\input{./figures/TikZit//rowtimesplustransposeprf2.tikz}}%
	\endpgfgraphicnamed

 =R_{j\times(a)+ i}
\]
  \end{proof}   

 \begin{proof}[Proof of \autoref{sawptranp}] 
    Note that  $i=a_{m-1}2^{m-1}+\cdots + a_{j_s}2^{j_s}+\cdots+a_{j_1}2^{j_1}+\cdots+ a_02^0, a_k \in \{0, 1\}, \bar{a}_k=a_k\oplus 1, \oplus$ is the modulo 2 addition, $j=a_{m-1}2^{m-1}+\cdots + \bar{a}_{j_s}2^{j_s}+\cdots+\bar{a}_{j_1}2^{j_1}+\cdots+ a_02^0$, $0\leq j_1< \cdots < j_s \leq m-1, 1\leq s \leq m$, so $j$ is different from $i$ exactly in the $j_1, \cdots, j_s$ places in their binary expansions. 
 Then $R_{i\leftrightarrow j}^T=$
   \[ 
	\beginpgfgraphicnamed{TikZit//rowswitchtransposeprf}
	\InputIfFileExists{TikZit//rowswitchtransposeprf.tikz}{}{\input{./figures/TikZit//rowswitchtransposeprf.tikz}}%
	\endpgfgraphicnamed
\]
\[ =R_{i\leftrightarrow j}
\]
where $b_s= \overline{(a_{j_1}\oplus a_{j_s})}, \cdots, b_k= \overline{(a_{j_1}\oplus a_{j_k})}.$
  \end{proof}   
  
   \begin{proof}[Proof of \autoref{anymatpresen}] 
  If $m\leq n$, then $C=\begin{pmatrix}
       K & O
          \end{pmatrix}_{2^m\times 2^n}=(\underbrace{10\cdots0}_{2^{n-m}})\otimes K$, where the  $O$ in $C$ is a zero matrix of size $2^m\times (2^n-2^m)$, 
  \[
 K= \begin{pmatrix}
        E_r &  \cdots &O \\
       \vdots & \ddots & \vdots \\
         O& \cdots  & O
          \end{pmatrix}_{2^m\times 2^m}
  \]
  Since $(\underbrace{10\cdots0}_{2^{n-m}})= %
	\beginpgfgraphicnamed{TikZit//vector1dotso}
	\InputIfFileExists{TikZit//vector1dotso.tikz}{}{\input{./figures/TikZit//vector1dotso.tikz}}%
	\endpgfgraphicnamed
,$  and $K$ can be represented by sequential composition of $2^m-r$ row multiplication elementary matrices (multiplying 0) of size $2^m\times 2^m$ whose diagrammatic representation is shown in Theorem \ref{rowitimesatm} as
  $$
	\beginpgfgraphicnamed{TikZit//rowitimes0}
	\InputIfFileExists{TikZit//rowitimes0.tikz}{}{\input{./figures/TikZit//rowitimes0.tikz}}%
	\endpgfgraphicnamed

  $$
  
  Therefore, C can now be represented by string diagrams.
  
   If $m> n$, then $C=\begin{pmatrix}
       K^{\prime} \\
        O
          \end{pmatrix}_{2^m\times 2^n}=
              \begin{pmatrix}
       1 \\
        0\\
        \vdots\\
        0
          \end{pmatrix}_{2^{m-n}\times 1}\otimes  K^{\prime}$, where the  $O$ in $C$ is a zero matrix of size $(2^m-2^n)\times 2^n$, 
  \[
 K^{\prime}= \begin{pmatrix}
        E_r &  \cdots &O \\
       \vdots & \ddots & \vdots \\
         O& \cdots  & O
          \end{pmatrix}_{2^n\times 2^n}
  \]
  Since $  \begin{pmatrix}
       1 \\
        0\\
        \vdots\\
        0
          \end{pmatrix}_{2^{m-n}\times 1}= %
	\beginpgfgraphicnamed{TikZit//vector1dotso2}
	\InputIfFileExists{TikZit//vector1dotso2.tikz}{}{\input{./figures/TikZit//vector1dotso2.tikz}}%
	\endpgfgraphicnamed
,$  and $K^{\prime}$ can be represented by sequential composition of $2^n-r$ row multiplication (by 0) elementary matrices of size $2^n\times 2^n$ whose diagrammatic representation is shown in Theorem \ref{rowitimesatm}, $C$ can now be represented by diagrams.  
          
          To summarise, $C$ can always be represented by a ZX diagram,  also each elementary matrix can be represented by a ZX diagram, therefore, if we reverse the procedures from $A$ to $C$, we then get the diagrammatic representation of $A$.
   \end{proof}   
   
\section*{Details for \autoref{match}}   
Since
\[
 SU(2)= \left\{\begin{pmatrix}
a & -\overline{b} \\
 b & \overline{a} \\
  \end{pmatrix} : a, b \in \mathbb{C}, |a|^2+|b|^2=1\right\},
  \]
  we can assume that 
$a=\cos \theta e^{i\alpha}, b=\sin \theta e^{i\beta},  0 \leq \theta \leq \frac{\pi}{2}, \alpha, \beta \in [0, 2\pi)$.
 
 Return to the matchgate $G(A,B)$, we assume that 
 \[
  A=   \begin{pmatrix}
    p  & q     \\
 r & s \\
  \end{pmatrix} =   \begin{pmatrix}
    \cos \alpha e^{i\sigma}  &  -\sin\alpha e^{-i\tau}    \\
  \sin \alpha e^{i\tau}  &  \cos\alpha e^{-i\sigma}    \\
  \end{pmatrix} 
 \]
 
  \[
  B=   \begin{pmatrix}
  w & x \\
 y & z \\
  \end{pmatrix} =   \begin{pmatrix}
    \cos \beta e^{i\psi}  &  -\sin\beta e^{-i\phi}    \\
  \sin\beta e^{i\phi}  &  \cos\beta e^{-i\psi}    \\
  \end{pmatrix} 
 \]
where $0 \leq \alpha, \beta \leq \frac{\pi}{2}, \sigma, \tau, \psi, \phi \in [0, 2\pi)$.

Now we can represent a matchgate in ZX diagram.
First assume $\alpha \neq \frac{\pi}{2},  \beta\neq \frac{\pi}{2}$. Then
    \[
 G(A,B)= \begin{pmatrix}
    \cos \alpha e^{i\sigma}  & 0 & 0 & -\sin\alpha e^{-i\tau}     \\
    0 &  \cos \beta e^{i\psi}  &  -\sin\beta e^{-i\phi} & 0   \\
    0 &  \sin\beta e^{i\phi}  &  \cos\beta e^{-i\psi}   & 0   \\
 \sin \alpha e^{i\tau} & 0 & 0 & \cos\alpha e^{-i\sigma}  \\
  \end{pmatrix}    \overset{R_{0\times(-e^{i(\tau-\sigma)}\tan\alpha )+ 3}}{\longrightarrow}
  \]
 \[  \begin{pmatrix}
    \cos \alpha e^{i\sigma}  & 0 & 0 & -\sin\alpha e^{-i\tau}     \\
    0 &  \cos \beta e^{i\psi}  &  -\sin\beta e^{-i\phi} & 0   \\
    0 &  \sin\beta e^{i\phi}  &  \cos\beta e^{-i\psi}   & 0   \\
0 & 0 & 0 & \sec\alpha e^{-i\sigma}  \\
  \end{pmatrix}   \overset{C_{0\times(e^{-i(\tau+\sigma)}\tan\alpha )+ 3}}{\longrightarrow}
  \]
  \[
  \begin{pmatrix}
    \cos \alpha e^{i\sigma}  & 0 & 0 & 0    \\
    0 &  \cos \beta e^{i\psi}  &  -\sin\beta e^{-i\phi} & 0   \\
    0 &  \sin\beta e^{i\phi}  &  \cos\beta e^{-i\psi}   & 0   \\
 0 & 0 & 0 & \sec\alpha e^{-i\sigma}  \\
  \end{pmatrix}  \overset{R_{1\times(-e^{i(\phi-\psi)}\tan\beta )+ 2}}{\longrightarrow}
 \]
  \[
  \begin{pmatrix}
    \cos \alpha e^{i\sigma}  & 0 & 0 & 0    \\
    0 &  \cos \beta e^{i\psi}  &  -\sin\beta e^{-i\phi} & 0   \\
    0 &  0 &  \sec\beta e^{-i\psi}   & 0   \\
 0 & 0 & 0 & \sec\alpha e^{-i\sigma}  \\
  \end{pmatrix}  \overset{C_{1\times(e^{-i(\phi+\psi)}\tan\beta )+ 2}}{\longrightarrow}
 \]
  \[
  \begin{pmatrix}
    \cos \alpha e^{i\sigma}  & 0 & 0 & 0    \\
    0 &  \cos \beta e^{i\psi}  &  0 & 0   \\
    0 &  0 &  \sec\beta e^{-i\psi}   & 0   \\
 0 & 0 & 0 & \sec\alpha e^{-i\sigma}  \\
  \end{pmatrix} =  \begin{pmatrix}
    \cos \alpha e^{i\sigma}  & 0     \\
    0 &   \sec\beta e^{-i\psi}     \\
  \end{pmatrix} \otimes \begin{pmatrix}
   1 & 0     \\
    0 &   \sec\alpha\cos\beta e^{i(\psi-\sigma)}     \\
  \end{pmatrix} 
 \]
  \[
 = cos \alpha e^{i\sigma}\begin{pmatrix}
      1  & 0     \\
    0 &  \sec\alpha \sec\beta e^{-i(\psi+\sigma)}      \\
  \end{pmatrix} \otimes \begin{pmatrix}
   1 & 0     \\
    0 &   \sec\alpha\cos\beta e^{i(\psi-\sigma)}     \\
  \end{pmatrix} =  cos \alpha e^{i\sigma} X \otimes Y
 \]

Therefore,   
\[
 G(A,B)= R_{0\times(e^{i(\tau-\sigma)}\tan\alpha )+ 3}R_{1\times(e^{i(\phi-\psi)}\tan\beta )+ 2}(cos \alpha e^{i\sigma} X \otimes Y) C_{1\times(-e^{-i(\phi+\psi)}\tan\beta )+ 2}C_{0\times(-e^{-i(\tau+\sigma)}\tan\alpha )+ 3}
\]

 Let $a=e^{i(\tau-\sigma)}\tan\alpha=\frac{r}{p}, b=e^{i(\phi-\psi)}\tan\beta=\frac{y}{w}, c=-e^{-i(\tau+\sigma)}\tan\alpha=\frac{q}{p}, d=-e^{-i(\phi+\psi)}\tan\beta=\frac{x}{w}$. Then $R_{0\times(e^{i(\tau-\sigma)}\tan\alpha )+ 3}R_{1\times(e^{i(\phi-\psi)}\tan\beta )+ 2}=R_{1\times(e^{i(\phi-\psi)}\tan\beta )+ 2}R_{0\times(e^{i(\tau-\sigma)}\tan\alpha )+ 3}$ can be represented as follows:
 
    \[
	\beginpgfgraphicnamed{elementarysimply}
	\InputIfFileExists{elementarysimply.tikz}{}{\input{./figures/elementarysimply.tikz}}%
	\endpgfgraphicnamed

      \]    
     \[
	\beginpgfgraphicnamed{elementarysimply2}
	\InputIfFileExists{elementarysimply2.tikz}{}{\input{./figures/elementarysimply2.tikz}}%
	\endpgfgraphicnamed

      \]       
Let $e= \sec\alpha \sec\beta e^{-i(\psi+\sigma)}  =\frac{1}{pw},  f=\sec\alpha\cos\beta e^{i(\psi-\sigma)}= \frac{w}{p}.$  Then $ G(A,B)$ can be represented as
        \[
	\beginpgfgraphicnamed{elementarysimplydoubled2v2}
	\InputIfFileExists{elementarysimplydoubled2v2.tikz}{}{\input{./figures/elementarysimplydoubled2v2.tikz}}%
	\endpgfgraphicnamed

      \]   
  Also it can be checked by computational basis that 
      
  \[
	\beginpgfgraphicnamed{element2matricessim}
	\InputIfFileExists{element2matricessim.tikz}{}{\input{./figures/element2matricessim.tikz}}%
	\endpgfgraphicnamed

      \]   
   Therefore  $G(A,B)$ can be represented in ZX as: 
      \[
	\beginpgfgraphicnamed{elementarysimplydoubled3}
	\InputIfFileExists{elementarysimplydoubled3.tikz}{}{\input{./figures/elementarysimplydoubled3.tikz}}%
	\endpgfgraphicnamed
=  %
	\beginpgfgraphicnamed{elementarysimplydoubledsim}
	\InputIfFileExists{elementarysimplydoubledsim.tikz}{}{\input{./figures/elementarysimplydoubledsim.tikz}}%
	\endpgfgraphicnamed
= \quad %
	\beginpgfgraphicnamed{matchgatezxw}
	\InputIfFileExists{matchgatezxw.tikz}{}{\input{./figures/matchgatezxw.tikz}}%
	\endpgfgraphicnamed

      \] 
The circuit form of the matchgate:
\begin{center}
\includegraphics[height=5cm]{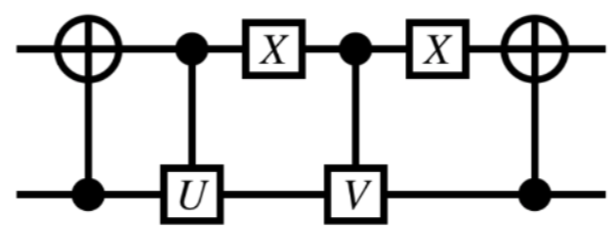} 
\end{center}
The ZW form of the matchgate:
\begin{center}
\includegraphics[height=5cm]{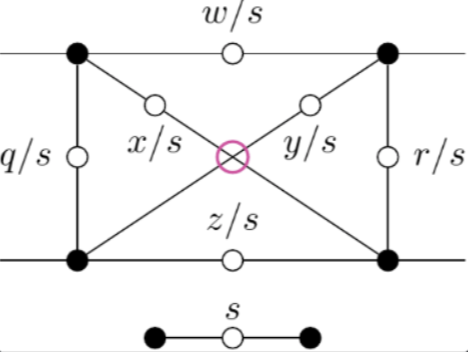} 
\end{center}
      
%



\end{document}

%% file: tikzstyles.tex
\tikzstyle{env}=[copoint,regular polygon rotate=0,minimum width=0.2cm, fill=black]

\definecolor{zx_red}{RGB}{232, 165, 165}
\definecolor{zx_green}{RGB}{216, 248, 216}

\tikzstyle{every picture}=[baseline=-0.25em]
\tikzstyle{dotpic}=[scale=0.5]
\tikzstyle{diredges}=[every to/.style={diredge}]
\tikzstyle{dot graph}=[shorten <=-0.1mm,shorten >=-0.1mm,scale=0.6]
\tikzstyle{plot point}=[circle,fill=black,minimum width=2mm,inner sep=0]

\tikzstyle{braceedge}=[decorate,decoration={brace,amplitude=2mm,raise=-1mm}]
\tikzstyle{small braceedge}=[decorate,decoration={brace,amplitude=1mm,raise=-1mm}]
\tikzstyle{left hook arrow}=[left hook-latex]
\tikzstyle{right hook arrow}=[right hook-latex]

\tikzstyle{dtriangle}=[fill=yellow,draw=black,shape=isosceles triangle,shape border rotate=-90,isosceles triangle stretches=true,inner sep=0.8pt,minimum width=0.25cm,minimum height=2mm]
\tikzstyle{vtriang}=[fill=yellow,draw=black,shape=isosceles triangle,shape border rotate=180,isosceles triangle stretches=true,inner sep=0.8pt,minimum width=0.25cm,minimum height=2mm]
\tikzstyle{trigmc}=[fill=zx_green,draw=black,shape=isosceles triangle,shape border rotate=90,isosceles triangle stretches=true,inner sep=0.8pt,minimum width=0.3cm,minimum height=2mm]
\tikzstyle{vrt}=[fill=yellow,draw=black,shape=isosceles triangle,shape border rotate=0,isosceles triangle stretches=true,inner sep=0.8pt,minimum width=0.25cm,minimum height=2mm]
\tikzstyle{H box}=[rectangle,fill=yellow,draw=black,xscale=0.8,yscale=0.8, inner sep=0.6pt]
\tikzstyle{gn_phase}=[minimum size=1.2em, font={\footnotesize\boldmath}, shape=rectangle, rounded corners=0.5em, inner sep=0.2em, outer sep=-0.2em, scale=0.8, draw=black, fill=zx_green]
\tikzstyle{gbox}=[rectangle,fill=zx_green,draw=black, xscale=1.0, yscale=1.0, inner sep=1.pt, shape=rectangle]
\tikzstyle{rbox}=[rectangle,fill=red,draw=black,xscale=1.0,yscale=1.0, inner sep=1.pt]
\tikzstyle{rn_phase}=[{gn_phase}, fill=zx_red, draw=black]
\tikzstyle{zhbx}=[rectangle,fill=white,draw=black,xscale=1.0,yscale=1.0, inner sep=1.6pt]
\tikzstyle{newh}=[rectangle,fill=yellow,draw=black,xscale=2.0,yscale=2.0, inner sep=1.6pt]
\tikzstyle{triangle}=[fill=yellow,draw=black,shape=isosceles triangle,shape border rotate=90,isosceles triangle stretches=true,inner sep=0.8pt,minimum width=0.25cm,minimum height=2mm]
\tikzstyle{rtriang}=[shape=isosceles triangle, fill=yellow, draw=black, isosceles triangle stretches=true, inner sep=0.8pt, minimum width=0.25cm, minimum height=2mm]
\tikzstyle{utriang}=[rtriang, shape=isosceles triangle, fill=yellow, draw=black, shape border rotate=90]
\tikzstyle{dtriang}=[rtriang, shape=isosceles triangle, fill=yellow, draw=black, shape border rotate=-90]
\tikzstyle{box}=[draw=black, shape=rectangle, fill={white}, minimum size=.95em, inner sep=0.15em, scale=0.85, font={\footnotesize}]

\tikzstyle{bn}=[circle,fill=black,draw=black,scale=.4]
\tikzstyle{wn}=[circle,fill=white,draw=black,scale=.6]
\tikzstyle{dn}=[circle,fill=none,draw=gray]
\tikzstyle{bspider}=[fill=black,draw=black,scale=1,shape=isosceles triangle,shape border rotate=-90,isosceles triangle stretches=true,inner sep=1pt,minimum width=0.4cm,minimum height=3mm]
\tikzstyle{dbspider}=[fill=black,draw=black,scale=1,shape=isosceles triangle,shape border rotate=90,isosceles triangle stretches=true,inner sep=1pt,minimum width=0.4cm,minimum height=3mm]
\tikzstyle{L}=[rectangle,shape=rectangle,fill=green,draw=black]

\tikzstyle{Z dot}=[inner sep=0mm, minimum size=2mm, shape=circle, draw=black, fill={rgb,255: red,221; green,255; blue,221}]
\tikzstyle{Z phase dot}=[minimum size=5mm, font={\footnotesize\boldmath}, shape=rectangle, rounded corners=2mm, inner sep=0.2mm, outer sep=-2mm, scale=0.8, draw=black, fill={rgb,255: red,221; green,255; blue,221}]
\tikzstyle{X dot}=[Z dot, shape=circle, draw=black, fill={rgb,255: red,255; green,136; blue,136}]
\tikzstyle{X phase dot}=[Z phase dot, fill={rgb,255: red,255; green,136; blue,136}, font={\footnotesize\boldmath}]
\tikzstyle{hadamard edge}=[-, dashed, dash pattern=on 2pt off 0.5pt, thick, draw={rgb,255: red,68; green,136; blue,255}]

\tikzstyle{black dot}=[inner sep=0.7mm,minimum width=0pt,minimum height=0pt,fill=black,draw=black,shape=circle]

\tikzstyle{dot}=[black dot]
\tikzstyle{smalldot}=[inner sep=0.4mm,minimum width=0pt,minimum height=0pt,fill=black,draw=black,shape=circle]
\tikzstyle{white dot}=[dot,fill=white]
\tikzstyle{antipode}=[white dot,inner sep=0.3mm,font=\footnotesize]
\tikzstyle{smallwhitedot}=[smalldot,fill=white]
\tikzstyle{alt white dot}=[white dot,label={[xshift=3.07mm,yshift=-0.05mm,font=\footnotesize]left:$*$}]
\tikzstyle{gray dot}=[dot,fill=gray!40!white]
\tikzstyle{smallgraydot}=[smalldot,fill=gray!40!white]
\tikzstyle{box vertex}=[draw=black,rectangle]
\tikzstyle{small box}=[box vertex,fill=white]
\tikzstyle{whitebg}=[fill=white,inner sep=2pt]
\tikzstyle{graph state vertex}=[sg vertex,fill=black]

\tikzstyle{wide copoint}=[fill=white,draw=black,shape=isosceles triangle,shape border rotate=90,isosceles triangle stretches=true,inner sep=1pt,minimum width=1.5cm,minimum height=5mm]
\tikzstyle{wide point}=[fill=white,draw=black,shape=isosceles triangle,shape border rotate=-90,isosceles triangle stretches=true,inner sep=1pt,minimum width=1.5cm,minimum height=4mm]
\tikzstyle{very wide copoint}=[fill=white,draw=black,shape=isosceles triangle,shape border rotate=-90,isosceles triangle stretches=true,inner sep=1pt,minimum width=2.5cm,minimum height=4mm]
\tikzstyle{very wide empty copoint}=[draw=black,shape=isosceles triangle,shape border rotate=-90,isosceles triangle stretches=true,inner sep=1pt,minimum width=2.5cm,minimum height=4mm]
\tikzstyle{symm}=[ultra thick,shorten <=-1mm,shorten >=-1mm]

\tikzstyle{square box}=[rectangle,fill=white,draw=black,minimum height=5mm,minimum width=5mm,font=\small]
\tikzstyle{square gray box}=[rectangle,fill=gray!30,draw=black,minimum height=6mm,minimum width=6mm]
\tikzstyle{copoint}=[regular polygon,regular polygon sides=3,draw=black,scale=0.75,inner sep=-0.5pt,minimum width=7mm,fill=white]
\tikzstyle{point}=[regular polygon,regular polygon sides=3,draw=black,scale=0.75,inner sep=-0.5pt,minimum width=7mm,fill=white,regular polygon rotate=180]
\tikzstyle{gray point}=[point,fill=gray!40!white]
\tikzstyle{gray copoint}=[copoint,fill=gray!40!white]

\newcommand{\edgearrow}{{\arrow[black]{>}}}
\newcommand{\edgetick}{{\arrow[black,scale=0.7,very thick]{|}}}

\tikzstyle{diredge}=[->]
\tikzstyle{rdiredge}=[<-]
\tikzstyle{medium diredge}=[->]

\tikzstyle{short diredge}=[->]
\tikzstyle{halfedge}=[-)]
\tikzstyle{other halfedge}=[(-]
\tikzstyle{freeedge}=[(-)]
\tikzstyle{white edge}=[line width=5pt,white]
\tikzstyle{tick}=[postaction=decorate,decoration={markings, mark=at position 0.5 with \edgetick}]
\tikzstyle{small map edge}=[|-latex, gray!60!blue, shorten <=0.9mm, shorten >=0.5mm]
\tikzstyle{thick dashed edge}=[very thick,dashed,gray!40]
\tikzstyle{map edge}=[|-latex,very thick, gray!40, shorten <=1mm, shorten >=0.5mm]
\tikzstyle{tickedge}=[postaction=decorate,
  decoration={markings, mark=at position 0.5 with \edgetick}]
\tikzstyle{dirtickedge}=[postaction=decorate,
  decoration={markings, mark=at position 0.5 with \edgetick},
  decoration={markings, mark=at position 0.85 with \edgearrow}]
\tikzstyle{dirdoubletickedge}=[postaction=decorate,
  decoration={markings, mark=at position 0.4 with \edgetick},
  decoration={markings, mark=at position 0.6 with \edgetick},
  decoration={markings, mark=at position 0.85 with \edgearrow}]

\makeatletter
\newcommand{\boxshape}[3]{%
\pgfdeclareshape{#1}{
\inheritsavedanchors[from=rectangle] 
\inheritanchorborder[from=rectangle]
\inheritanchor[from=rectangle]{center}
\inheritanchor[from=rectangle]{north}
\inheritanchor[from=rectangle]{south}
\inheritanchor[from=rectangle]{west}
\inheritanchor[from=rectangle]{east}
\backgroundpath{
\southwest \pgf@xa=\pgf@x \pgf@ya=\pgf@y
\northeast \pgf@xb=\pgf@x \pgf@yb=\pgf@y

\@tempdima=#2
\@tempdimb=#3

\pgfpathmoveto{\pgfpoint{\pgf@xa - 5pt + \@tempdima}{\pgf@ya}}
\pgfpathlineto{\pgfpoint{\pgf@xa - 5pt - \@tempdima}{\pgf@yb}}
\pgfpathlineto{\pgfpoint{\pgf@xb + 5pt + \@tempdimb}{\pgf@yb}}
\pgfpathlineto{\pgfpoint{\pgf@xb + 5pt - \@tempdimb}{\pgf@ya}}
\pgfpathlineto{\pgfpoint{\pgf@xa - 5pt + \@tempdima}{\pgf@ya}}
\pgfpathclose
}
}}

\boxshape{NEbox}{0pt}{8pt}
\boxshape{SEbox}{0pt}{-8pt}
\boxshape{NWbox}{8pt}{0pt}
\boxshape{SWbox}{-8pt}{0pt}
\makeatother

\tikzstyle{map}=[draw,shape=NEbox,inner sep=7pt]
\tikzstyle{mapdag}=[draw,shape=SEbox,inner sep=7pt]
\tikzstyle{maptrans}=[draw,shape=SWbox,inner sep=7pt]
\tikzstyle{mapconj}=[draw,shape=NWbox,inner sep=7pt]

\tikzstyle{probs}=[shape=semicircle,fill=gray!40!white,draw=black,shape border rotate=180,minimum width=1.2cm]

\tikzstyle{arrs}=[-latex,font=\small,auto]
\tikzstyle{arrow plain}=[arrs]
\tikzstyle{arrow dashed}=[dashed,arrs]
\tikzstyle{arrow bold}=[very thick,arrs]
\tikzstyle{arrow hide}=[draw=white!0,-]
\tikzstyle{arrow reverse}=[latex-]
\tikzstyle{cdnode}=[]

\tikzstyle{gn}=[dot,fill=zx_green,minimum width=0.25cm,inner sep=0pt]
\tikzstyle{rno}=[dot,fill=red,inner sep=0pt,minimum width=0.25cm]
\tikzstyle{rn}=[dot,fill=zx_red,inner sep=0pt,minimum width=0.25cm]

\tikzstyle{rc}=[dot,thick,fill=white,draw = red,minimum width=0.3cm,inner sep=0pt]
\tikzstyle{gc}=[dot,thick,fill=white,draw=zx_green,inner sep=0pt,minimum width=0.3cm]
\tikzstyle{bc}=[dot,thick,fill=white,draw= blue,minimum width=0.3cm]

\tikzstyle{label}=[circle,fill=white,minimum width=0.3cm]

\tikzstyle{clocklabel}=[dot,fill=yellow,draw=black,font=\tiny,inner sep=0.75pt]

\tikzstyle{rsn}=[circle split,draw,fill=red,font=\tiny,inner sep=0.75pt]
\tikzstyle{gsn}=[circle split,draw,fill=zx_green,font=\tiny,inner sep=0.75pt]
\tikzstyle{bsn}=[circle split,draw,fill=blue,font=\tiny,inner sep=0.75pt]

\tikzstyle{rsc}=[circle split,thick,draw= red,draw,fill=white,font=\tiny,inner sep=0.75pt]
\tikzstyle{gsc}=[circle split,thick,draw=zx_green,draw,fill=white,font=\tiny,inner sep=0.75pt]
\tikzstyle{bsc}=[circle split,thick,draw= blue,draw,fill=white,font=\tiny,inner sep=0.75pt]

\tikzstyle{cnot}=[fill=white,shape=circle,inner sep=-1.4pt]
\tikzstyle{wire label}=[font=\tiny, auto]

%% file: defs.tex
\tikzstyle{cdiag}=[matrix of math nodes, row sep=3em, column sep=3em, text height=1.5ex, text depth=0.25ex,inner sep=0.5em]
\tikzstyle{arrow above}=[transform canvas={yshift=0.5ex}]
\tikzstyle{arrow below}=[transform canvas={yshift=-0.5ex}]


%% file: TikZit/newhadamard.tikz
\begin{tikzpicture}
	\begin{pgfonlayer}{nodelayer}
		\node [style=newh] (0) at (0, 0) {};
		\node [style=none] (1) at (0, 0.5) {};
		\node [style=none] (2) at (0, -0.5) {};
	\end{pgfonlayer}
	\begin{pgfonlayer}{edgelayer}
		\draw (1.center) to (2.center);
	\end{pgfonlayer}
\end{tikzpicture}

%% file: TikZit/triangle.tikz
\begin{tikzpicture}
	\begin{pgfonlayer}{nodelayer}
		\node [style=none] (0) at (0, 0.5) {};
		\node [style=triangle] (1) at (0, 0) {};
		\node [style=none] (2) at (0, -0.5) {};
	\end{pgfonlayer}
	\begin{pgfonlayer}{edgelayer}
		\draw (0.center) to (2.center);
	\end{pgfonlayer}
\end{tikzpicture}

%% file: TikZit/triangleinv.tikz
\begin{tikzpicture}
	\begin{pgfonlayer}{nodelayer}
		\node [style=none] (0) at (0.25, 0.25) {-{\scriptsize1}};
		\node [style=triangle] (1) at (0, 0) {};
		\node [style=none] (2) at (0, -0.5) {};
		\node [style=none] (3) at (0, 0.5) {};
	\end{pgfonlayer}
	\begin{pgfonlayer}{edgelayer}
		\draw (3.center) to (2.center);
	\end{pgfonlayer}
\end{tikzpicture}

%% file: TikZit/singleredpi.tikz
\begin{tikzpicture}
	\begin{pgfonlayer}{nodelayer}
		\node [style=none] (0) at (0, 0.5) {};
		\node [style=rn_phase] (1) at (0, 0) {$\pi$};
		\node [style=none] (2) at (0, -0.5) {};
	\end{pgfonlayer}
	\begin{pgfonlayer}{edgelayer}
		\draw (0.center) to (2.center);
	\end{pgfonlayer}
\end{tikzpicture}

%% file: TikZit/Id.tikz
\begin{tikzpicture}
	\begin{pgfonlayer}{nodelayer}
		\node [style=none] (1) at (0.5, 0.3) {};
		\node [style=none] (2) at (0.5, -0.3) {};
		\node [style=none] (3) at (0.5, -0.5) {};
		\node [style=none] (4) at (0.5, 0.5) {};
	\end{pgfonlayer}
	\begin{pgfonlayer}{edgelayer}
		\draw (1.center) to (2.center);
	\end{pgfonlayer}
\end{tikzpicture}

%% file: TikZit/cap.tikz
\begin{tikzpicture}
	\begin{pgfonlayer}{nodelayer}
		\node [style=none] (0) at (0, -0) {};
		\node [style=none] (1) at (1, -0) {};
	\end{pgfonlayer}
	\begin{pgfonlayer}{edgelayer}
		\draw [bend left=90, looseness=1.50] (0.center) to (1.center);
	\end{pgfonlayer}
\end{tikzpicture}

%% file: TikZit/cup.tikz
\begin{tikzpicture}
	\begin{pgfonlayer}{nodelayer}
		\node [style=none] (0) at (0, 0.5) {};
		\node [style=none] (1) at (1, 0.5) {};
	\end{pgfonlayer}
	\begin{pgfonlayer}{edgelayer}
		\draw [bend right=90, looseness=1.50] (0.center) to (1.center);
	\end{pgfonlayer}
\end{tikzpicture}

%% file: TikZit/0matrix2by2.tikz
\begin{tikzpicture}
	\begin{pgfonlayer}{nodelayer}
		\node [style=rn_phase] (0) at (-0.25, 0) {$\pi$};
		\node [style=none] (1) at (0, 0.5) {};
		\node [style=none] (2) at (0, -0.5) {};
	\end{pgfonlayer}
	\begin{pgfonlayer}{edgelayer}
		\draw (2.center) to (1.center);
	\end{pgfonlayer}
\end{tikzpicture}

%% file: TikZit/simplified2by2mt.tikz
\begin{tikzpicture}
	\begin{pgfonlayer}{nodelayer}
		\node [style=gbox] (0) at (0, 0) {$k$};
		\node [style=none] (1) at (0, -0.5) {};
		\node [style=none] (2) at (0, 0.5) {};
	\end{pgfonlayer}
	\begin{pgfonlayer}{edgelayer}
		\draw (2.center) to (1.center);
	\end{pgfonlayer}
\end{tikzpicture}